\documentclass[11pt]{article}
\usepackage{graphicx}
\usepackage{subcaption}
\usepackage{todonotes}
\usepackage{mathtools}
\usepackage{soul}
\usepackage{amsthm}
\usepackage{amsmath,amssymb}
\usepackage{systeme}

\DeclareMathOperator*{\argmin}{arg\,min}
\usepackage{gensymb}
\usepackage[ruled,vlined]{algorithm2e}
\SetKwComment{Comment}{/* }{ */}
\SetKwInOut{Input}{input}\SetKwInOut{Output}{output}

\graphicspath{{./figs/}}
\usepackage[absolute]{textpos}

\usepackage{lineno}
\usepackage{fullpage,cite}
\usepackage[top=0.8in, bottom=0.9in, left=0.9in, right=0.9in]{geometry}

\usepackage[pdftex, plainpages = false, pdfpagelabels,
                 bookmarks=false,
                 bookmarksopen = true,
                 bookmarksnumbered = true,
                 breaklinks = true,
                 linktocpage,
                 pagebackref,
                 colorlinks = true,  % was true
                 linkcolor = blue,
                 urlcolor  = cyan,
                 citecolor = red,
                 anchorcolor = green,
                 hyperindex = true,
                 hyperfigures
                 ]{hyperref}

\def\bbR{\mathbb{R}}
\def\calD{\mathcal{D}}
\def\calB{\mathcal{B}}
\def\le{\text{LE}}
\def\inp{\text{IN}}
\def\live{\text{LIVE}}
\newtheorem{theorem}{Theorem}[section]
\newtheorem{lemma}[theorem]{Lemma}
\newtheorem{corollary}[theorem]{Corollary}
\newtheorem{observation}{Observation}
\numberwithin{algocf}{section}
\title{Dynamic 3D Convex Hulls Revisited and Applications\thanks{This paper was published in SODA 2026~\cite{ref:WangDy26}. After publication, an error was discovered in the main algorithm, which invalidates the primary result of the paper. See the appendix in the last page for an explanation.}
}
\author{Haitao Wang\thanks{Kahlert School of Computing,
University of Utah, Salt Lake City, UT 84112, USA. {\tt haitao.wang@utah.edu}}
}
\date{}

\begin{document}

% \begin{textblock}{5}(1,0.5)
% \noindent\Large {\bf APPENDIX}
% \end{textblock}

\maketitle

\vspace{-0.3in}
\begin{abstract}
Chan [J.~ACM, 2010] developed a data structure to maintain the convex hull of a dynamic set of points in 3D under insertions and deletions to answer certain queries on the convex hulls. The algorithm has been slightly improved by Kaplan, Mulzer, Roditty, Seiferth, and Sharir~[Discret. Comput. Geom, 2020] and by Chan [Discret. Comput. Geom, 2020], without altering the main algorithmic framework. The current best result supports each insertion in $O(\log^2 n)$ amortized time, each deletion in $O(\log^4 n)$ amortized time, and each extreme query in $O(\log^2 n)$ worst-case time (along with other query types). 
These results have numerous applications, notably in dynamic Euclidean nearest neighbor searching in 2D. In the dual setting, the problem becomes maintaining the lower envelope of a dynamic set of 3D planes for vertical ray-shooting queries. By developing randomized vertical shallow cutting algorithms for general distance functions, Kaplan, Mulzer, Roditty, Seiferth, and Sharir~[Discret. Comput. Geom, 2020] and Liu [SIAM J. Comput., 2022] extended Chan’s framework to maintain the lower envelope of a dynamic set of general 3D surfaces. The best known result in this setting achieves $O(\log^2 n)$ amortized expected time for insertions, $O(\log^4 n)$ amortized expected time for deletions, and $O(\log^2 n)$ worst-case time for vertical ray-shooting queries. As an immediate application, dynamic nearest neighbor searching under general distance functions (e.g., the $L_p$ metric or additively-weighted Euclidean distance) can be solved.

In this paper, we revisit Chan’s algorithmic framework and propose a modified version that reduces the deletion time to $O(\log^3 n\log \log n)$, while retaining the $O(\log^2 n)$ insertion time at the cost of an increased query time of $O(\log^3 n / \log \log n)$. Our result is particularly appealing in scenarios where the overall running time is dominated by the operation with the highest time complexity; in such cases, our approach offers an improvement of roughly a logarithmic factor over prior work based on Chan’s original framework. This improvement translates into faster algorithms for several fundamental problems, including maintaining the dynamic 2D bichromatic closest pair, computing convex layers of 3D points, maintaining the minimum spanning tree of a dynamic 2D point set, and maintaining dynamic disk graph connectivity, among others.

%If the edge weight is define as the minimum distance between two disks, then the previous best algorithm runs in $O(n\log^4 n)$ expected time using Chan's algorithmic framework. Replacing Chan's framework with ours, we immediately obtain an improved algorithm of $O(n\log^3 n\log\log n)$ expected time. If the edge weight of the disk graphs is defined as the distance of the disk centers,  

In particular, we consider the problem of computing shortest paths in weighted disk graphs. 
The previous best algorithm, due to An, Oh, and Xue [SoCG 2025], runs in $O(n\log^4 n)$ time. We present a new algorithm that achieves an improved expected running time of $O(n\log^3 n\log\log n)$. Note that our algorithm is not merely the result of plugging our new dynamic data structure into the previous approach. Rather, we propose a different method, in which the new data structure plays a critical supporting role.
\end{abstract}


\section{Introduction}
\label{sec:intro}
We consider the problem of maintaining the convex hull of a dynamic set of points in $\bbR^3$, supporting insertions, deletions, and extreme point queries; that is, given a query direction, return the most extreme vertex of the convex hull along that direction. The two-dimensional version of this problem has been extensively studied~\cite{ref:Overmarsma81,ref:ChanDy01,ref:BrodalDy00,ref:BrodalDy02}, culminating in a data structure by Brodal and Jacob~\cite{ref:BrodalDy02} that supports both updates and queries in $O(\log n)$ time.

The 3D version of the problem is significantly more challenging. In 2006, Chan~\cite{ref:ChanA10} presented the first data structure with polylogarithmic update and query times. Specifically, each insertion requires $O(\log^3 n)$ amortized expected time, each deletion takes $O(\log^6 n)$ amortized expected time, and each extreme point query can be answered in $O(\log^2 n)$ worst-case time. With a minor modification, Kaplan, Mulzer, Roditty, Seiferth, and Sharir~\cite{ref:KaplanDy20} improved the deletion time by a logarithmic factor. Later, using a new deterministic algorithm for computing vertical shallow cuttings developed by Chan and Tsakalidis~\cite{ref:ChanOp16}, Chan~\cite{ref:ChanDy20} made the entire framework deterministic and further reduced both insertion and deletion times by another logarithmic factor. Specifically, the amortized times for insertion and deletion became $O(\log^2 n)$ and $O(\log^4 n)$, respectively, while maintaining the $O(\log^2 n)$ worst-case query time. It is worth noting that all these results follow the same algorithmic framework introduced in~\cite{ref:ChanA10}.

These results have numerous applications. For instance, in the dual setting, the problem becomes the {\em dynamic lower envelope problem}: maintaining the lower envelope of a dynamic set of planes in $\mathbb{R}^3$ to support vertical ray-shooting queries, i.e., given a vertical query line $\ell$, report the intersection of $\ell$ with the lower envelope. Consequently, the time bounds described above apply directly to this problem. In fact, the algorithms in \cite{ref:ChanA10,ref:ChanDy20,ref:KaplanDy20} are developed directly on this problem. 
Moreover, via a well-known lifting transformation~\cite{ref:deBergCo08}, Euclidean nearest neighbor queries among a dynamic set of 2D points can be reduced to the 3D dynamic lower envelope problem, allowing the above results to be applied immediately. However, Chan’s original result does not extend directly to nearest neighbor searching under more general distance functions, such as the commonly used $L_p$ metric (for $p \in [1, \infty]$) and additively-weighted Euclidean distance.

Under general distance functions, the nearest neighbor searching problem reduces to the {\em dynamic surface lower envelope problem}: maintaining the lower envelope of a dynamic set of totally defined continuous bivariate functions of constant description complexity in $\mathbb{R}^3$, supporting vertical ray-shooting queries. A key component of Chan’s data structure~\cite{ref:ChanA10,ref:ChanDy20} is an algorithm for computing vertical shallow cuttings~\cite{ref:ChanOp16}, which, however, only applies to planes.
To handle general distance functions, Kaplan, Mulzer, Roditty, Seiferth, and Sharir~\cite{ref:KaplanDy20} developed a randomized vertical shallow cutting algorithm for general 3D surfaces. Incorporating their algorithm into Chan’s framework yields a solution to the dynamic surface lower envelope problem. Liu~\cite{ref:LiuNe22} later improved the shallow cutting algorithm, leading to the following performance: each insertion takes $O(\log^2 n)$ amortized expected time, each deletion takes $O(\log^4 n)$ amortized expected time, and each ray-shooting query takes $O(\log^2 n)$ worst-case time, assuming that the lower envelope of any subset of the functions has linear complexity. These complexities match those of Chan’s original results for planes, although the algorithms are randomized and the time bounds are in expectation.
As a result, dynamic nearest neighbor searching under general distance functions can now be addressed directly. This has many applications, particularly because the $L_p$ metric and the additively-weighted Euclidean metric are widely used in practice.

\subsection{Our results}

In many applications that rely on the above results, the overall runtime is dominated by the operation with the highest time complexity—namely, insertions, deletions, or queries. In such cases, reducing this bottleneck leads to an improvement in the total running time of the algorithm. This observation motivates the present work. Our goal is to reduce the highest time complexity, which is the $O(\log^4 n)$ time for deletions.

We revisit Chan’s data structure and propose a nontrivial modification to its framework. Our modified data structure improves the highest time complexity of Chan’s original structure by nearly a logarithmic factor. Specifically, it achieves $O(\log^2 n)$ amortized time for insertions, $O(\log^3 n \log \log n)$ amortized time for deletions, and $O(\log^3 n / \log \log n)$ worst-case time for queries. Intuitively, one may interpret this as shifting a $\log n / \log \log n$ factor from the deletion time to the query time.

Moreover, since the data structures for the surface case developed by Kaplan, Mulzer, Roditty, Seiferth, and Sharir~\cite{ref:KaplanDy20} and Liu~\cite{ref:LiuNe22} follow Chan’s framework, our modification applies there as well. The resulting data structure achieves the same time bounds as above, with the caveat that both the insertion and deletion times are expected rather than deterministic.

As dynamic nearest neighbor searching reduces to the dynamic lower envelope problem, our results yield improved data structures for nearest neighbor queries with the same time complexities. For the Euclidean metric, our data structure is deterministic; for general distance functions, it is randomized.

\paragraph{Applications.}
Our new data structures lead to improvements in algorithms and data structures for a variety of problems. In general, if a previous algorithm relies on Chan’s data structure~\cite{ref:ChanA10,ref:ChanDy20} or the surface case data structure~\cite{ref:KaplanDy20,ref:LiuNe22}, and if the overall time complexity is dominated by the most expensive operation among insertions, deletions, and queries, then replacing the original data structure with ours as a black box can likely improve the algorithm's time complexity by a $\log n / \log \log n$ factor.

Problems of this nature include dynamic 2D bichromatic closest pair, computing convex layers of 3D points, maintaining the minimum spanning tree of a dynamic 2D point set, computing minimum Euclidean bichromatic matchings, maintaining dynamic disk graph connectivity, and computing spanners for disk graphs, among others (see Section~\ref{sec:app} for the details).

\paragraph{Shortest paths in disk graphs.}
In particular, we consider the problem of computing shortest paths in {\em proximity graphs} defined on a set of disks in the plane. Each disk corresponds to a vertex, and two disks are connected by an edge if their minimum distance is at most a prescribed parameter $r \geq 0$; when $r = 0$, the graph becomes the standard disk graph. If edge weights are defined as the minimum distance between disks, the best known algorithm, due to Kaplan, Katz, Saban, and Sharir~\cite{ref:KaplanTh23}, runs in $O(n \log^4 n)$ expected time, relying on the dynamic surface lower envelope data structure~\cite{ref:KaplanDy20,ref:LiuNe22}. By replacing this data structure with ours, we immediately obtain an improved algorithm with expected running time $O(n \log^3 n \log \log n)$.

If, instead, edge weights are defined as the distance between disk centers, the best known algorithm, which is also from~\cite{ref:KaplanTh23}, runs in $O(n \log^6 n)$ expected time, again using the dynamic surface lower envelope structure. For the special case of disk graphs (i.e., $r = 0$), a recent result by An, Oh, and Xue~\cite{ref:AnSi25} at SoCG 2025 gives a new algorithm with running time $O(n \log^4 n)$.

By applying our new data structure to the algorithm of~\cite{ref:KaplanTh23}, we obtain an improved algorithm for the center-distance case with expected time $O(n \log^5 n \log \log n)$. Beyond that, we propose an entirely new algorithm, which runs in $O(n \log^3 n \log \log n)$ expected time, even improving upon the result of An, Oh, and Xue~\cite{ref:AnSi25} that works for the disk graph only. Our algorithm is based on several novel observations, is relatively simple, and may serve as a useful framework for solving shortest path problems in other variants of disk graphs.

\paragraph{Dynamic halfspace range reporting and $\boldsymbol{k}$-nearest neighbor searching.}
Our technique can also be used to improve existing results for certain query problems whose output size includes a factor of $k$, such as $k$-nearest neighbor searching and halfspace range reporting. Based on an observation by Chan~\cite{ref:ChanA10}, his dynamic framework can be extended to support {\em $k$-lowest-planes} queries (i.e., given a vertical query line, report the $k$ lowest planes intersecting it; note that vertical ray-shooting on the lower envelope is the special case $k = 1$). Using Chan’s original framework~\cite{ref:ChanA10,ref:ChanDy20}, such queries can be supported with $O(\log^2 n)$ amortized insertion time, $O(\log^4 n)$ amortized deletion time, and $O((k + \log n)\log n)$ query time.
By replacing Chan’s algorithm with ours, we achieve the same amortized insertion time of $O(\log^2 n)$, an improved amortized deletion time of $O(\log^3 n \log \log n)$, and a query time of $O((k + \log n)\log^2 n / \log \log n)$. 

Furthermore, using Chan’s framework, de Berg and Staals~\cite{ref:BergDy23} developed a data structure with query time $O(k + \log^2 n / \log \log n)$, amortized insertion time $O(\log^{3+\varepsilon} n)$, and amortized deletion time $O(\log^{5+\varepsilon} n)$, for any constant $\varepsilon > 0$.
We obtain a new result that achieves query time $O(k + \log^3 n / \log \log n)$, amortized insertion time $O(\log^2 n \log \log n)$, and amortized deletion time $O(\log^4 n \log \log n)$. These results apply to both Euclidean $k$-nearest neighbor searching among a dynamic set of 2D points and halfspace range reporting among a dynamic set of 3D points. 

In addition, similar improvements can be obtained for $k$-nearest neighbor searching under general distance functions.

\paragraph{Outline.} The rest of the paper is organized as follows. After introducing some notation and basic concepts in Section~\ref{sec:pre}, we present our dynamic data structure in Section~\ref{sec:main}. Section~\ref{sec:spdisk} gives our shortest path algorithm for proximity graphs of disks. In Section~\ref{sec:rangereport}, we discuss the dynamic $k$-lowest-planes queries and halfspace range reporting.

\section{Preliminaries}
\label{sec:pre}

For any two points $p$ and $q$ in the plane, we use $\lVert p - q \rVert$ to denote their Euclidean distance.

Let $H$ be a set of $n$ planes in $\mathbb{R}^3$. The {\em level} of a point $p \in \mathbb{R}^3$ is defined as the number of planes in $H$ that lie strictly below $p$. For an integer $k$, the {\em $(\leq k)$-level} consists of all points in $\mathbb{R}^3$ whose level is at most $k$. A {\em $(k, K)$-shallow cutting} for $H$ is a collection $\Xi$ of pairwise interior-disjoint cells $\triangle$ of constant complexity such that their union covers the $(\leq k)$-level of $H$, and the interior of each cell $\triangle$ intersects at most $K$ planes of $H$. The {\em size} of $\Xi$ is the number of cells it contains. For each cell $\triangle \in \Xi$, let $H_{\triangle}$ denote the set of planes in $H$ that intersect the interior of $\triangle$; this set is called the {\em conflict list} of $\triangle$.

We say that a cell $\triangle \in \Xi$ is {\em vertical} if it is a downward (semi-unbounded) tetrahedron consisting of all points that lie vertically below some trapezoid. A {\em vertical shallow cutting} is one in which every cell is vertical. For any $k \in [1, n]$, there exists a $(k, O(k))$-vertical shallow cutting of size $O(n/k)$~\cite{ref:MatousekRe92,ref:ChanRa00}; such a cutting can be computed in $O(n \log n)$ time~\cite{ref:ChanOp16}.

Let $\le(H)$ denote the lower envelope of $H$. Throughout this paper, a {\em vertical line} in $\mathbb{R}^3$ refers to a line parallel to the $z$-axis.

In the following discussion, all stated time complexities are deterministic unless explicitly noted as {\em expected time}, and are worst-case unless {\em amortized time} is specified.

\section{The dynamic lower envelope data structure}
\label{sec:main}

In this section, we present our data structure for the dynamic lower envelope problem, obtained by modifying Chan’s data structure~\cite{ref:ChanA10,ref:ChanDy20}. We begin by briefly reviewing Chan’s algorithm in Section~\ref{sec:chanreview}, and then describe our modifications and new result in Section~\ref{sec:new}. In Section~\ref{sec:surface}, we extend our result to the surface lower envelope setting. Finally, Section~\ref{sec:app} presents several immediate applications of our new data structure.

\subsection{A review of Chan's algorithm}
\label{sec:chanreview}

%We first briefly review Chan's algorithm~\cite{ref:ChanDy20,ref:ChanA10,ref:KaplanDy20} and then explain our modifications. 
Let $H$ be a set of $n$ planes in $\bbR^3$. We use $\calD(H)$ to denote Chan's data structure on $H$; it is constructed by a preprocesing algorithm which we call ``{\tt Chan-Preprocess($H$)}''. We explain the algorithm below. 

Let $b,c,c'$ be three constants that will be decided later. The algorithm has $t=\log_bn$ iterations. Initially, let $H_0=H$. Each $i$-th iteration, $1\leq i\leq t$, computes a $(n/b^i,cn/b^i)$-vertical shallow cutting $\Xi_i$ for $H_{i-1}$ with at most $c'b^i$ cells. Then, let $H_i$ be $H_{i-1}$ by removing those ``bad'' planes $h\in H_{i-1}$ that intersects more than $2cc't$ cells of $\Xi_1,\Xi_2,\ldots,\Xi_i$; note that these bad planes are still in $H_0,H_1,\ldots,H_{i-1}$. Furthermore, for each cell $\triangle\in \Xi_i$, we compute the conflict list $(H_i)_{\triangle}$ and initialize a number $\kappa_{\triangle}=0$. Note that the planes of $H_{i-1}\setminus H_i$ do not appear in any conflict list of $\Xi_i$ but still exist in the conflict lists of $\Xi_1,\Xi_2,\ldots,\Xi_{i-1}$. We construct a data structure for the last cutting $\Xi_t$ so that given a vertical line $\ell$, the cell of $\Xi_t$ intersecting $\ell$ can be computed efficiently. Since all cells of $\Xi_t$ are vertical, this can be done by first projecting $\Xi_t$ onto the $xy$-plane and then construct a planar point location data structure~\cite{ref:KirkpatrickOp83,ref:EdelsbrunnerOp86}. Then, given $\ell$, the cell of $\Xi_t$ intersecting $\ell$ can be computed in $O(\log n)$ time using a point location query. Note that the conflict list of each cell of $\Xi_t$ has $O(1)$ planes.
Finally, let $\calB(H)=H\setminus H_t$ (i.e., the set of bad planes); we call {\tt Chan-Preprocess($\calB(H)$)} recursively (until $|\calB(H)|=O(1)$).

It can be shown that $|\calB(H)|\leq n/2$~\cite{ref:ChanDy20}. Hence, {\tt Chan-Preprocess($H$)} will recur $O(\log n)$ times (sometimes we say it has $O(\log n)$ {\em recursion levels}). We say that all the cuttings $\Xi_1,\ldots, \Xi_t$ form a {\em substructure} of $\calD(H)$. Hence, {\tt Chan-Preprocess($H$)} constructs $O(\log n)$ substructures, each having $O(\log n)$ vertical shallow cuttings. All these substructures constitute $\calD(H)$. As shown in \cite{ref:ChanDy20}, there exist constants $b,c,c'$ so that all cuttings $\Xi_1,\ldots, \Xi_t$ in the first substructure can be constructed in $O(n\log n)$ time by using the shallow cutting algorithm of Chan and Tsakalidis~\cite{ref:ChanOp16}. As such, the total time of {\tt Chan-Preprocess($H$)} for constructing $\calD(H)$ satisfies the recurrence $T(n)=T(n/2)+O(n\log n)$, which solves to $T(n)=O(n\log n)$.  

We introduce some notation that will be used later. For each substructure $D$ of $\calD(H)$, we use $\inp(D)$ to denote the subset of the planes of $H$ that are initially used to construct $D$. After $D$ is constructed, a set $\calB(\inp(D))$ of bad planes are generated due to the construction of $D$; for simplicity, we let $\calB(D)=\calB(\inp(D))$. In addition, we let $\live(D)=\inp(D)\setminus \calB(D)$, which are called ``live'' planes in \cite{ref:ChanA10}. 
%we also say that these planes are ``stored'' in $D$. 
For example, the above described algorithm is for constructing the first level substructure 
$D$, which consists of the  shallow cuttings $\Xi_1,\ldots, \Xi_t$; we have $\inp(D)=H$, 
$\calB(D)=\calB(H)$, and $\live(D)=H\setminus \calB(D)$, which is actually $H_t$. Note that each plane is live in exactly one substructure $D$ of $\calD(H)$.
%and the plane appears in $O(\log n)$ conflict lists of the shallow cuttings of $D$. 

\paragraph{Insertions.}
To insert a plane $h$ to $H$, we use a function, called {\tt Chan-Insert($H,h$)}. It simply inserts $h$ to $\calB(H)$ recursively. When $|\calB(H)|$ reaches $3n/4$, we rebuild $\calD(H)$. Since there are at least $\Omega(n)$ updates before a rebuild occurs after the last rebuild, the amortized insertion time satisfies the recurrence $T(n)=T(3n/4)+O(P(n)/n)$, where $P(n)=O(n\log n)$ is the time for {\tt Chan-Preprocess($H$)}. The recurrence solves to $T(n)=O(\log^2 n)$. 

\paragraph{Deletions.}
To delete a plane $h$ from $H$, we call a function, {\tt Chan-Delete($H,h$)}, which works as follows. For any cell $\triangle\in \Xi_i$, for all $1\leq i\leq t$, such that $h\in (H_i)_{\triangle}$, we increment $\kappa_{\triangle}$ by one. If $\kappa_{\triangle}\geq n/b^{i+1}$, then for all planes $h'\in (H_i)_{\triangle}$ that are not in $\calB(H)$, re-insert $h'$ to $H$ by calling {\tt Chan-Insert($H,h'$)}, which effectively inserts $h'$ to $\calB(H)$. Finally, if $h$ is also in $\calB(H)$, then we call {\tt Chan-Delete$(\calB(H),h)$} recursively. 
We will analyze the deletion time later. 

\paragraph{Queries.}
Given a vertical line $\ell$, our goal is to compute its intersection with the lower envelope $\le(H)$ of $H$; equivalently, we seek the lowest plane of $H$ intersecting $\ell$. Instead of using Chan’s original query algorithm~\cite{ref:ChanDy20,ref:ChanA10}, we adopt the approach from Kaplan, Mulzer, Roditty, Seiferth, and Sharir~\cite{ref:KaplanDy20}, which can be easily extended to the more general surface setting discussed in Section~\ref{sec:surface}.

For each substructure $D$ of $\mathcal{D}(H)$, we locate the cell $\triangle$ in its last cutting that intersects $\ell$, using a point location data structure in $O(\log n)$ time. We then inspect the conflict list of $\triangle$: among all planes in the conflict list that have not been deleted, we select the one that intersects $\ell$ at the lowest point. Since the conflict list contains only $O(1)$ planes, this step takes $O(1)$ time.

Repeating this process for all $O(\log n)$ substructures yields $O(\log n)$ candidate planes in total, with an overall runtime of $O(\log^2 n)$. The plane among these candidates whose intersection with $\ell$ is lowest is returned as the final answer. We refer to~\cite{ref:KaplanDy20} for the correctness of this algorithm. The total query time is $O(\log^2 n)$. We denote this query procedure as {\tt KMRSS-Query($H,\ell$)}.

\paragraph{Deletion time.}
We now analyze the deletion time, which will help motivate our new approach. Suppose we delete a plane $h$ from $H$. According to {\tt Chan-Preprocess($H$)}, for each substructure $D$ of $\calD(H)$, $h$ appears in the conflict lists of at most $O(\log n)$ cells of all shallow cuttings of $D$. Therefore, in each substructure, there are $O(\log n)$ cutting cells $\triangle$ such that $\kappa_{\triangle}$ increments due to the deletion of $h$. It takes $n/b^{i+1}$ increments of $\kappa_{\triangle}$ to meet the condition $\kappa_{\triangle}\geq n/b^{i+1}$ so that $|(H_i)_{\triangle}|\leq cn/b^i$ planes are re-inserted to $\calB(H)$. Hence, deleting $h$ from each substructure of $\calD(H)$ triggers amortized $O(\log n)$ re-insertions. As such, the amortized deletion time satisfies the recurrence $T(n)=T(3n/4)+O(\log n)\cdot I(3n/4)=T(3n/4)+O(\log^3 n)$, where $I(n)$ is the insertion time. Hence, the recurrence solves to $T(n)=O(\log^4 n)$.

An alternative way to view the deletion time is that since deleting $h$ from each substructure of $\calD(H)$ triggers $O(\log n)$ amortized re-insertions and there are $O(\log n)$ substructures, the total amortized number of re-insertions is $O(\log^2 n)$. 
To make the analysis of our new algorithm in Section~\ref{sec:new} easier, we turn this amortized analysis into worst-case by slightly modifying the deletion algorithm (the preliminary version of \cite{ref:ChanA10} in SODA 2006 actually used this strategy). Specifically, instead of re-inserting all planes of $(H_i)_{\triangle}$ after $\kappa_{\triangle}$ reaches $n/b^{i+1}$, we re-insert a constant number of planes $(H_i)_{\triangle}$ whenever $\kappa_{\triangle}$ increments (see the preliminary version of \cite{ref:ChanA10} for a concrete constant number; roughly speaking, we use a constant so that the conflict list $(H_i)_{\triangle}$ becomes empty no later than when $\kappa_{\triangle}$ reaches $n/b^{i+1}$). In this way, deleting $h$ from each substructure of $\calD(H)$ triggers at most $O(\log n)$ re-insertions and the total number of re-insertions for all substructures is $O(\log^2 n)$ in the worst case. From now on, {\tt Chan-Delete($H,h$)} refers to this worst-case version. 

%\paragraph{Remark.}

\subsection{Our new algorithm}
\label{sec:new}

We begin with introducing an observation that motivates our new approach. Let $D_1,D_2,\ldots,D_l$ denote the substructures of $\calD(H)$ in the order they are constructed. Specifically, $D_1$ is constructed using $H$, $D_2$ is constructed using $\calB(H)$, etc. 
Recall that $|\calB(D)|\leq |\inp(D)|/2$ for each substructure $D$. 
%For each $D_i$, let $H'_i$ denote the subset of $H$ that is used to construct $D_i(H)$. According to Chan-Preprocess($H$), $H'_{i+1}=\calB(H_i')$, $|H_{i+1}'|\leq |H_i'|/2$, and $|H'_{t}|=O(1)$. 
For the last substructure $D_l$, we have $|\inp(D_l)|=O(1)$, and thus we simply let $D_l$ be the list of planes of $\inp(D_l)$. 

For each $1\leq j\leq l$, for each plane $h\in\live(D_j)$, we define $j$ as the {\em substructure-level} of $h$. 
If the substructure-level of a plane $h$ is $j$, then our observation is that $h$ cannot appear in the conflict lists of any substructures after $D_j$ and therefore we only need to delete $h$ from the first $j$ substructures. In the above deletion time analysis, we have used $O(\log n)$ as the upper bound for the number of substructures from which we need to delete a plane. While this upper bound is correct, it is a bit ``over-estimated''. Indeed, since $|\calB(D_j)|\leq |\inp(D)_j|/2$ for each substructure $D_j$ of $\calD(H)$ and $\calB(D_j)$ becomes $\inp(D_{j+1})$ for the next substructure $D_{j+1}$, 
$|\inp(D_j)|$ for $j=1,2,\ldots,l$ are geometrically decreasing, and therefore, most planes have small substructure-level numbers. Specifically, since at least half of the planes of $H$ are in $\live(D_1)$, at least half of the planes have substructure-level equal to $1$. Similarly, at least a quarter of the planes have substructure-levels equal to $2$. In general, at least $n/2^j$ planes have substructure-levels equal to $j$, for any $1\leq j\leq l$. Therefore, the total sum of the substructure-levels of all planes of $H$ is at most $2n$ (this resembles to the time analysis of the linear-time heap construction algorithm). Hence, the {\em average} substructure-level per plane is only $2$, rather than $\log n$. 

If we were to use $2$ instead of $\log n$ in the time analysis of {\tt Chan-Delete($H,h$)}, then the number of re-insertions triggered by a deletion would be $O(\log n)$ and consequently the amortized deletion time would be $O(\log^3 n)$. 
Our goal is to modify Chan’s algorithm in a way that converts this average-case analysis into a worst-case guarantee as much as possible.
It should be noted that although we are ultimately satisfied with amortized bounds, we still need to ensure that the time complexity of any sequence of updates remains bounded. 
The above ``average'' substructure-level argument alone does not imply a worst-case guarantee for any sequence of updates.

\paragraph{Necessity of the change.}  
We begin with an intuitive explanation of why a modification to Chan’s algorithm appears necessary to achieve improved deletion time, despite the favorable average-case behavior. While the \emph{average} substructure-level is small, it might be still possible to construct a sequence of updates in which each deleted plane has substructure level equal to or close to $l = \Theta(\log n)$. 
%This suggests that the worst-case deletion time remains $O(\log^4 n)$ unless the algorithm is fundamentally changed.

Consider the following scenario. Suppose we first delete a plane $h$ that is stored in the last substructure $D_l$ of $\calD(H)$, and has substructure level $l$. The deletion of $h$ may trigger a series of re-insertions that cause some substructure to be rebuilt, during which a plane $h'$, originally stored in a different substructure, is moved to $D_l$. If we next delete $h'$, it once again has substructure level $l$. By repeating this pattern, it might be possible to construct a sequence of $\Omega(n)$ deletions in which each deleted plane has substructure-level close to $l$.

This example demonstrates that it might be difficult to obtain a better deletion time analysis within the existing framework of Chan~\cite{ref:ChanA10,ref:ChanDy20}. To achieve an improved bound, a structural modification to Chan’s algorithm seems necessary.

\paragraph{Our solution.}
In the following, we describe our modifications to Chan’s algorithm. Our approach is motivated by the intuitive worst-case deletion sequence discussed above. The main issue is that re-insertions may trigger the rebuilding of high-level substructures in $\calD(H)$, causing subsequently deleted planes to again have high substructure-levels.
Our remedy is simple: redirect the insertions to new locations.

We now partition $H$ into two subsets $H^+$ and $H^-$, called the {\em main} and {\em auxiliary subsets} of $H$, respectively, such that $|H^-|< \alpha\cdot n/\log n$ for some constant $\alpha$ to be determined later (recall that $n=|H|$). Initially, $H^+=H$ and $H^-=\emptyset$.
We use $\calD^*(H)$ to denote our new data structure, which consists of $\calD(H^+)$ and $\calD^*(H^-)$. As such, $\calD^*(H)$ is defined recursively. Here, $\calD(H^+)$ refers to the data structure computed by calling {\tt Chan-Preprocess($H^+$)}. 
As the basic case, if $|H|=O(1)$, then $\calD^*(H)$ is simply the list of planes of $H$. Our key idea is that whenever we (re)-insert a plane, we insert it to the auxiliary subsets. The full details are described below.

\subsubsection{Insertions}
Our new insertion function, {\tt Insert($H,h$)}, is described in Algorithm~\ref{algo:insert}. We first call {\tt Insert$(H^-,h)$}, i.e., insert $h$ into $H^-$ recursively. If $|H^-| \geq \alpha n/\log n$, we perform a procedure {\tt Merge($H^-,H^+$)} to merge $H^-$ into $H^+$, after which $H^+=H$ and $H^-=\emptyset$, and $\calD^*(H)$ is ``reset'' to $\calD(H)$. 

\begin{algorithm}[h]
	\caption{Insert($H,h$)}
	\label{algo:insert}
	\SetAlgoNoLine
	%\KwIn{$P_{\alpha}$ and $H_{\alpha}=\{h_1,h_2,\ldots,h_n\}$}
	%\KwOut{An indirect solution $(P'_{\alpha},W_{\alpha})$} \BlankLine
	%\tcc{All indices below are understood modulo $n$.}
        Insert($H^-,h$)\;
        \If(\tcp*[f]{$n=|H|$}){$|H^-|\geq \alpha\cdot n/\log n$} 
        {
        Merge($H^-,H^+$)\tcp*[r]{merge $H^-$ into $H^+$} 
        }
\end{algorithm}

\begin{algorithm}[h]
	\caption{Merge($H^-,H^+$)}
	\label{algo:merge}
	\SetAlgoNoLine
	%\KwIn{$P_{\alpha}$ and $H_{\alpha}=\{h_1,h_2,\ldots,h_n\}$}
	%\KwOut{An indirect solution $(P'_{\alpha},W_{\alpha})$} \BlankLine
	%\tcc{All indices below are understood modulo $n$.}
        %$H^+\leftarrow H^+\cup H^-$ \tcp*[r]{$H^+=H$ after this line}         
        \eIf(\tcp*[f]{$n=|H|$}){$|H^-|+|\calB(H^+)|\geq 3n/4$}
        {
        Chan-Preprocess($H$)\;
        $H^+\leftarrow H$\;
        $H^-\leftarrow \emptyset$\;
        }
        {
        Merge($H^-,\calB(H^+)$)\tcp*[r]{merge $H^-$ into $\calB(H^+)$ recursively} 
        }        
\end{algorithm}

One way to implement the merge procedure is to insert every plane of $H^-$ into $D(H^+)$ one by one by calling {\tt Chan-Insert()}. Because {\tt Chan-Insert()} takes $O(\log^2 n)$ amortized time, the amortized time of our new {\tt Insert($H,h$)} satisfies the following recurrence: $T(n)=T(\alpha n/\log n)+O(\log^2 n)$, which solves to $T(n)=O(\log^3n/\log\log n)$. We can improve the time to $O(\log^2 n)$ by inserting the planes of $H^-$ into $D(H^+)$ all together in a ``batched'' fashion, as follows. 

Our merge procedure is described in Algorithm~\ref{algo:merge}, which runs recursively. 
%We start with setting $H^+=H^+\cup H^-$, implying that $H^+=H$. 
%Note that $\calB(H^+)$ is the set of ``bad'' planes produced by Chan-Preprocess($H^+$). 
%similar to $\calB(H)$ defined with respect to $H$ in Section~\ref{sec:chanreview}. 
If $|H^-|+|\calB(H^+)|\geq 3n/4$, then we call {\tt Chan-Preprocess($H$)} to rebuild $\calD(H)$, after which $\calD^*(H)$ is ``reset'' to $\calD(H)$, and we also set $H^+=H$ and $H^-=\emptyset$. 
Otherwise, we merge $H^-$ into $\calB(H^+)$ recursively, by calling {\tt Merge$(H^-,\calB(H^+))$}.  
%Finally, we set $H^-=\emptyset$.

\begin{lemma}\label{lem:insert-time}
The amortized insertion time is $O(\log^2 n)$.
\end{lemma}
\begin{proof}
We first analyze the time complexity of the merge procedure and then discuss the time of {\tt Insert($H,h$)}. 

Since $|\calB(H^+)|\leq n/2$ right after {\tt Chan-Preprocess($H^+$)}, each {\tt Chan-Preprocess($H^+$)} in the merge procedure is called after $\Omega(n)$ planes from $H^-$'s have been merged into $\calB(H^+)$. Since the procedure {\tt Chan-Preprocess($H^+$)} takes $O(n\log n)$ time, the amortized time of {\tt Chan-Preprocess($H^+$)} per plane of $H^-$ is $O(\log n)$. Hence, the amortized time of per plane of $H^-$ in the merge procedure satisfies the recurrence $T(n)=T(3n/4)+O(\log n)$, which solves to $T(n)=O(\log^2 n)$. However, we claim that the recursion depth is bounded by $O(\log\log n)$, which leads to $T(n)=O(\log n\log\log n)$. We prove the claim below. 

Indeed, recall that $\calD(H^+)$ consists of $O(\log n)$ substructures $D_j$, $1\leq j\leq O(\log n)$, and $|\inp(D_j)|\leq n/2^{j-1}$ holds for each $j$. 
Recall that $|H^-|=\alpha n/\log n$ when {\tt Merge$(H^-,H^+)$} is invoked in {\tt Insert($H,h$)}. 
Let $j$ be the largest index so that $|H^-|=\alpha n/\log n\leq |\inp(D_j)|$. Then, $\alpha n/\log n\leq n/2^{j-1}$ since $|\inp(D_j)|\leq n/2^{j-1}$. Hence, $j= O(\log\log n)$. Observe that the recursion of the merge procedure will stop no later than when it reaches the substructure of $D_j$, that is, the recursion has depth at most $O(\log \log n)$. 

In summary, the amortized time of the merge procedure per plane of $H^-$ is $O(\log n\log\log n)$. Hence, the amortized time of our new insertion algorithm {\tt Insert($H,h$)} satisfies the recurrence $T(n)=T(\alpha n/\log n)+O(\log n\log\log n)$, which solves to $T(n)=O(\log^2 n)$.
\end{proof}

\subsubsection{Deletions}
\label{sec:deletion}
To delete a plane $h$ from $H$, we call a function {\tt Delete$(H,h)$}. As $H^+$ and $H^-$ form a partition of $H$, $h$ is either in $H^+$ or in $H^-$. In the latter case, we call {\tt Delete$(H^-,h)$}, i.e., delete $h$ from $H^-$ recursively. In the former case, we call {\tt Chan-Delete($H^+,h$)}, i.e., apply the Chan's deletion algorithm, but with our new insertion algorithm, i.e., whenever we need to re-insert a plane $h'$, we call {\tt Insert($H,h'$)}, which effectively inserts $h'$ into $H^-$. See Algorithms~\ref{algo:delete} and \ref{algo:chandelete} for the pseudocode of {\tt Delete$(H,h)$} and {\tt Chan-Delete($H^-,h$)}, respectively. We will argue later that the amortized deletion time is $O(\log^3 n\log\log n)$.

\begin{algorithm}[h]
	\caption{Delete($H,h$)}
	\label{algo:delete}
	\SetAlgoNoLine
	%\KwIn{$P_{\alpha}$ and $H_{\alpha}=\{h_1,h_2,\ldots,h_n\}$}
	%\KwOut{An indirect solution $(P'_{\alpha},W_{\alpha})$} \BlankLine
	%\tcc{All indices below are understood modulo $n$.}
        \eIf{$h\in H^-$}
        {
        Delete($H^-,h$) \tcp*[r]{delete $h$ from $H^-$ recursively} 
        }
        {
        Chan-Delete($H^+,h$)\tcp*[r]{call Chan's deletion algorithm} 
        }
\end{algorithm}

\begin{algorithm}[h]
	\caption{Chan-Delete($H^+,h$)}
	\label{algo:chandelete}
	\SetAlgoNoLine
	%\KwIn{$P_{\alpha}$ and $H_{\alpha}=\{h_1,h_2,\ldots,h_n\}$}
	%\KwOut{An indirect solution $(P'_{\alpha},W_{\alpha})$} \BlankLine
	\tcc{Chan's deletion algorithm but using our new insertion algorithm}
        \For{\em each cell $\triangle$ of the shallow cuttings in the first substructure of $\calD(H^+)$ whose conflict list contains $h$}
        {
         %increment $\kappa_{\triangle}$\;
          remove $O(1)$ planes from the conflict list of $\triangle$\; %\tcp*[r]{we use a constant so that the conflict list becomes empty when $\kappa_{\triangle}$ reaches $n/b^{i+1}$, assuming that $\triangle$ is in a shallow cutting $\Xi_i$ discussed in Section~\ref{sec:chanreview}; see the preliminary version of \cite{ref:ChanA10} for details} 
          \For{\em each removed plane $h'$ that is still in $H^+$}
          {
            mark $h'$ as ``removed'' from $H^+$\;
            Insert($H,h'$)\tcp*[r]{this effectively inserts $h'$ to $H^-$} 
          }
        }
        \If{$h\in \calB(H^+)$}
        {
          Chan-Delete($\calB(H^+),h$)\tcp*[r]{delete $h$ recursively from $\calB(H^+)$}
        }
\end{algorithm}

\subsubsection{Queries}
%Before analyzing the deletion time, we discuss the query algorithm. 
Given a vertical line $\ell$, our query algorithm {\tt Query($H,\ell$)} is given in Algorithm~\ref{algo:query}. We first apply {\tt KMRSS-Query($H^+,\ell$)}, i.e., use the query algorithm described in Section~\ref{sec:chanreview} to find the lowest plane at $\ell$ among all planes in $H^+$. Then, we call the query algorithm recursively on $H^-$, which finds the lowest plane of $H^-$ at $\ell$. Among the two planes, we return the lower one at $\ell$. The correctness follows a similar argument as the KMRSS-Query algorithm~\cite{ref:KaplanDy20}. Since {\tt KMRSS-Query($H^+,\ell$)} takes $O(\log^2 n)$ time, the total query time satisfies the recurrence $T(n)=T(\alpha n/\log n)+O(\log^2 n)$, which solves to $T(n)=O(\log^3 n/\log\log n)$. 

\begin{algorithm}[h]
	\caption{Query($H,\ell$)}
	\label{algo:query}
	\SetAlgoNoLine
	%\KwIn{$P_{\alpha}$ and $H_{\alpha}=\{h_1,h_2,\ldots,h_n\}$}
	%\KwOut{An indirect solution $(P'_{\alpha},W_{\alpha})$} \BlankLine
	%\tcc{All indices below are understood modulo $n$.}
        $h^+\leftarrow$ KMRSS-Query($H^+,\ell$)\; 
        $h^-\leftarrow$ Query($H^-,\ell$)\; 
        \Return the lower one of $h^+$ and $h^-$ at $\ell$\;
\end{algorithm}

\subsubsection{Deletion time}
We now analyze the deletion time. 
Suppose we delete a plane $h$ from $H$. In each {\tt Chan-Delete()} procedure, excluding the recursive calls, $h$ appears in $O(\log n)$ conflict lists and therefore triggers $O(\log n)$ re-insertions as discussed in Section~\ref{sec:chanreview}. To analyze the deletion time, the key is to establish an upper bound on the number of {\tt Chan-Delete()} calls incurred by a single deletion.

Recall from Section~\ref{sec:chanreview} that the total sum of substructure-levels of all planes in $H^+$ within $\calD(H^+)$ is at most $2|H^+|$. This implies that if we were to delete all planes in $H^+$, and if $\calD(H^+)$ remained unchanged throughout the deletions, then the total number of \texttt{Chan-Delete()} calls would be at most $2|H^+|$, and thus the \emph{average} number of \texttt{Chan-Delete()} calls per deletion would be at most $2$.

However, in our framework, re-insertions triggered by deletions are directed into $H^-$, causing the size of $H^-$ to grow over time. Once $H^-$ exceeds a threshold, it is merged into $H^+$, thereby reconstructing $\calD(H^+)$ and resetting the structure. Therefore, a key question is: what is the \emph{maximum} number of deletions from $H^+$ that can occur before $H^-$ must be merged into $H^+$? Intuitively, the larger this number, the smaller the average number of \texttt{Chan-Delete()} calls per deletion, which leads to better amortized deletion time.

We start with the following lemma.

\begin{lemma}\label{lem:levelsum}
Suppose that $\calD(H^+)$ has just been constructed by a call to {\tt Chan-Preprocess($H^+$)}. Then, for any subset of $K$ planes from $H^+$ with $K \geq n / \log^c n$ for any constant $c$, the total sum of substructure-levels of the planes in the subset within $\calD(H^+)$ is $O(K \log \log n)$.
\end{lemma}
\begin{proof}
Recall that $\calD(H^+)$ consists of $O(\log n)$ substructures $D_j$, $1 \leq j \leq O(\log n)$, with $|\inp(D_j)| \leq n / 2^{j-1}$. Let $g$ be the largest index such that
\[
\sum_{j \geq g} \frac{n}{2^{j-1}} \geq K' = \frac{n}{\log^c n}
\]
for some constant $c$. Then, $g = O(\log \log n)$. By definition, $\sum_{j > g} \frac{n}{2^{j-1}} < K'$. Recall that $\calB(D_g) = \bigcup_{j > g} \live(D_j)$, and that the substructure-level of every plane $h \in \live(D_j)$ is $j$. Since $|\live(D_j)|\leq |\inp(D_j)|\leq n/2^{j-1}$, the total sum of substructure-levels of all planes in $\calB(D_g)$ is
\[
\sum_{j > g} |\live(D_j)| \cdot j \leq \sum_{j > g} \frac{n \cdot j}{2^{j-1}}.
\]

We now show that $\sum_{j > g} \frac{n \cdot j}{2^{j-1}} = O(g \cdot K)$. Indeed,
\[
\begin{split}
\sum_{j > g} \frac{n \cdot j}{2^{j-1}} 
&= \sum_{j > g} \frac{n \cdot g}{2^{j-1}} + \sum_{j > g} \frac{n \cdot (j - g)}{2^{j-1}} \\
&= g \cdot \sum_{j > g} \frac{n}{2^{j-1}} + \frac{n}{2^g} \cdot \sum_{j > g} \frac{j - g}{2^{j - g - 1}} \\
&= g \cdot \sum_{j > g} \frac{n}{2^{j-1}} + \frac{n}{2^g} \cdot \sum_{i > 0} \frac{i}{2^{i-1}} \\
&\leq g \cdot \sum_{j > g} \frac{n}{2^{j-1}} + \frac{n}{2^g} \cdot 4.
\end{split}
\]
Since $n / 2^g < \sum_{j > g} \frac{n}{2^{j-1}} < K' \leq K$, it follows that $\sum_{j > g} \frac{n \cdot j}{2^{j-1}} = O(g \cdot K)$.

Let $S$ be the subset of $K$ planes of $H^+$ in the lemma statement. We partition $S$ into two subsets: $S_1 = S \cap \calB(D_g)$ and $S_2 = S \setminus S_1$. For $S_1$, the total sum of substructure-levels is at most that of all planes in $\calB(D_g)$, which is $O(g \cdot K)$. For $S_2$, since each plane belongs to some $\live(D_j)$ with $j \leq g$, its substructure-level is at most $g$, and thus the total contribution from $S_2$ is also $O(g \cdot K)$.

In summary, the total sum of substructure-levels over all $K$ planes in $S$ is $O(g \cdot K) = O(K \log \log n)$, completing the proof of the lemma.
\end{proof}

\begin{corollary}\label{coro:levelsum}
Suppose that $\calD(H^+)$ has just been constructed by a call to {\tt Chan-Preprocess($H^+$)}. Then, 
for any subset of $K$ planes of $H^+$ with $K\leq K'= n/\log^c n$ for any constant $c$, the total sum of substructure-levels of all planes of the subset within $\calD(H^+)$ is $O(K'\log\log n)$.
\end{corollary}
\begin{proof}
Let $S$ be a subset of $K$ planes of $H^+$. We arbitrarily add $K'-K$ additional planes of $H^+$ to $S$ to obtain a new subset $S'$ with $|S'|=K'$. By Lemma~\ref{lem:levelsum}, the total sum of substructure-levels of all planes of $S'$ in $\calD(H^+)$ is $O(K'\log\log n)$. Since $S\subseteq S'$, the total sum of substructure-levels of $S$ is at most that of $S'$. The corollary thus follows. 
\end{proof}

We will use Lemma~\ref{lem:levelsum} in the following way. Suppose we perform a sequence of $K = n / \log^c n$ deletions from $H^+$, during which the data structure $\calD(H^+)$ remains unchanged. Then, by Lemma~\ref{lem:levelsum}, the total number of \texttt{Chan-Delete()} calls incurred by these deletions is $O(K \log \log n)$. Recall that deleting a plane of substructure-level $j$ may trigger $O(j)$ \texttt{Chan-Delete()} calls; thus, if the total sum of substructure-levels for a subset of planes is $M$, the total number of \texttt{Chan-Delete()} calls due to deletions of these planes is $O(M)$.
Since a single \texttt{Chan-Delete()} call can trigger $O(\log n)$ re-insertions, the total number of re-insertions during the entire sequence is $O(K \log n \log \log n)$. Consequently, the amortized number of re-insertions per deletion is only $O(\log n \log \log n)$.

In the following, we aim to analyze the more general setting where insertions are interspersed with deletions, and $\calD(H^+)$ may be rebuilt due to (re-)insertions. We first consider the case of deletions from $H^+$, and then extend the analysis to deletions from $H^-$.

%to $H^-$, if we set $H^-=c^*n/\log n$ for some constant $c^*$ , then before $H^-$ needs to merged into $H^+$, an amortized $K=n/(\log^2 n\log\log n))$ planes are deleted from $H^+$. After $H^-$ is merged into $H^+$, we can follow the same analysis. Therefore, we obtain that the amortized Chan-Delete calls is $O(n/\log^2 n/K)$, which is $O(\log\log n)$.

%The above only discuss the situation that there are $K$ consecutive deletions of the planes in $H^+$. If during the above sequence, there are also deletions happen to planes of $H^-$, then those deletions will only delay the merge of $H^-$ and $H^+$ and therefore it still holds that the amortized Chan-Delete calls of each deletion of the plane of $H^+$ is $O(\log\log n)$.

%The above analyzes the deletion-only case for deletions of $H^+$. Following the same analysis, the amortized Chan-Delete calls of each deletion of the plane of $H^-$ in the deletion-only case (i.e., no new insertion or re-insertion happen to $H^-$) is also $O(\log\log n)$.
%Next, we consider the general case where insertions and deletions are interspersed.

\paragraph{Deletions of $\boldsymbol{H^+}$.}
Note that $\calD(H^+)$ can only change as a result of the procedure {\tt Merge($H^-,H^+$)}. 
We consider the time interval (i.e., the sequence of updates) between two consecutive {\tt Merge($H^-,H^+$)}, and denote this interval by $\xi$. Depending on whether at least $K=n/(\log^3 n\log\log n)$ deletions of $H^+$ occur during $\xi$, we distinguish two cases. 

\begin{enumerate}
    \item 
If there are at least $K$ deletions, then since no {\tt Merge($H^-,H^+$)} happens in $\xi$, by Lemma~\ref{lem:levelsum}, there are $O(K\log n\log\log n)$ re-insertions triggered by all these deletions. Hence, the amortized number of re-insertions triggered by each deletion is $O(\log n\log\log n)$. For reference purpose, we call this type of deletions {\em type-1} deletions. 

\item 
If there are fewer than $K$ deletions during $\xi$, we call these deletions {\em type-2} deletions. In this case, by Corollary~\ref{coro:levelsum}, at most $O(K\log n\log\log n)=O(n/\log^2 n)$ re-insertions are triggered by these deletions. To simplify the analysis, we assume there are at most $n/\log^2 n$ re-insertions. 
%(which is indeed true for a sufficiently large $n$). 
As $|H^-|$ needs to reach $\alpha\cdot n/\log n$ to start a merge, if we set $\alpha=3$, other than the re-insertions triggered by the above deletions, there are at least $2n/\log n$ additional insertions to $H^-$ (let $S$ denote the set of these insertions) in the time interval $\xi$ (these insertions are actually all new insertions, not re-insertions). 

To account for the cost of the type-2 deletions, we assign each insertion of $S$ one ``credit''. We charge each insertion of $S$ a fractional cost of $1/\log n$ credits (the reason for this choice will be explained shortly). In total, the insertions in $S$ contribute at least $2n / \log^2 n$ credits. These are sufficient to cover the at most $n / \log^2 n$ re-insertions caused by the above type-2 deletions and still leave one \emph{extra credit} per re-insertion for future use.

Thus, the cost of handling type-2 deletions can be fully amortized over the insertions in $S$. Consequently, each deletion from $H^+$ triggers an amortized $O(\log n \log \log n)$ number of re-insertions.
\end{enumerate}

\paragraph{Deletions of $\boldsymbol{H^-}$.}
Deletions from $H^-$ can be analyzed in a similar manner, though a more refined charging argument is required. In particular, the \emph{extra credits} reserved earlier will play a crucial role in the analysis.

Recall that our data structure $\calD^*(H)$ is defined recursively, where each level of recursion defines a \emph{Chan-structure}. For example, $\calD(H^+)$ constitutes the Chan-structure at the first level. The total number of recursion levels follows the recurrence $T(n) = T(\alpha n / \log n) + 1$, which solves to $T(n) = O(\log n / \log \log n)$.
For the $i$-th recursion level, let $H_i^+$ and $H_i^-$ denote the main and auxiliary subsets, respectively, and define $H_i = H_i^+ \cup H_i^-$. For instance, $H_1^+ = H^+$, $H_1^- = H^-$, and $H_1 = H$. By construction, the size of $H_i$ satisfies
$\alpha^{i-1} n / \log^{i-1} n < |H_i| \leq \alpha^i n / \log^i n$.
Our previous analysis focused on deletions from $H_1^+$, corresponding to the first recursion level. We now extend the analysis to other levels. For simplicity, we focus on deletions from $H_2^+$ at the second level; the extension to higher levels follows inductively using the same argument.

Recall from our earlier analysis of the first-level structure $\calD(H^+)$ that each new insertion $h \in S$ is assigned one credit, of which only $1 / \log n$ is used at that level. The remaining credits are reserved for other recursion levels.
For each level, we keep certain amount of partial credits for $h$ so that the total partial credits of all levels are no more than one. Specifically, for the $i$-th level, we reserve $1 / \log m$ credits for $h$, where $m = \alpha^{i-1} n / \log^{i-1} n$ is the lower bound of $|H_i|$.
The total credit usage across all levels satisfies the recurrence
$T(n) = T(\alpha \cdot n / \log n) + 1 / \log n$,
which solves to $T(n) \leq 1$, assuming an appropriate constant in the base case.
Thus, assigning one credit to each insertion $h \in S$ suffices to pay for all levels of the recursive data structure $\calD^*(H)$.

Also recall that in our earlier analysis, we reserved one extra credit for each re-insertion triggered by a type-2 deletion. As with new insertions, we partition this extra credit into partial credits across the levels of $\calD^*(H)$ using the same scheme described above. Similarly, for each re-insertion triggered by a type-1 deletion from $H^+$, we assign one credit and partition it in the same manner across the recursion levels. 

In what follows, we argue that the partial credits reserved for the second recursion level of $\calD^*(H)$ are sufficient to cover the cost of the type-2 deletions from $H_2^+$. Let $m = |H_2|$. Note that $\calD(H^+)$ can only change as a result of the procedure {\tt Merge($H_2^-,H_2^+$)}. Consider the time interval $\xi$ between two consecutive \texttt{Merge}($H_2^-, H_2^+$) operations. As before, we distinguish two cases based on whether the number of deletions from $H_2^+$ during $\xi$ is at least $K = m/(\log^3 m \log \log m)$.

\begin{enumerate}
    \item 
If there are at least $K$ deletions during $\xi$, then by Lemma~\ref{lem:levelsum} (which, although proved for $H_1^+$, applies equally to every $H_i^+$ by the same argument), the total number of re-insertions triggered by these deletions is at most $O(K \log m \log \log m)$. Consequently, the amortized number of re-insertions for each deletion is $O(\log m \log \log m)$. We refer to these as \emph{type-1} deletions at the second level.

\item 
If there are fewer than $K$ deletions, we refer to them as {\em type-2} deletions. In this case, by Corollary~\ref{coro:levelsum}, at most $O(K\log m\log\log m)=O(m/\log^2 m)$ re-insertions triggered by these deletions. To simplify the analysis, we assume there are at most $m/\log^2 m$ re-insertions. Since a merge of $H_2^-$ into $H_2^+$ is triggered only when
$|H_2^-|$ reaches $\alpha\cdot m/\log m$ for $\alpha =3$, other than the re-insertions triggered by the type-2 deletions of $H^+_2$, there are at least $2m/\log m$ additional insertions to $H_2^+$ during the interval $\xi$. Let $S_2$ be the set of these insertions. 

Each insertion in $S_2$ falls into one of the following three categories:  
(1) a new insertion,  
(2) a re-insertion triggered by type-1 deletions of $H_1^+$, or  
(3) a re-insertion triggered by type-2 deletions of $H_1^+$.  
For each such insertion, we have reserved $1 / \log m$ partial credits specifically for this recursion level. Thus, the total number of partial credits available for this level is at least $2m / \log^2 m$. This is sufficient to pay for the at most $m / \log^2 m$ re-insertions triggered by type-2 deletions of $H_2^+$, while still assigning one extra credit to each of these re-insertions. As before, each extra credit is recursively partitioned into partial credits for use in higher recursion levels.

In this way, the cost of type-2 deletions of $H_2^+$ is effectively covered ``for free'' by the new insertions and the re-insertions triggered by deletions of $H_1^+$.  
Additionally, for each re-insertion triggered by type-1 deletions of $H_2^+$, we assign one credit and partition it across the recursion levels of $\calD^*(H)$ as done previously.
\end{enumerate}

The above analysis shows that each deletion from $H^+_2$ triggers $O(\log m \log\log m)$ amortized re-insertions, which is bounded by $O(\log n \log\log n)$ since $m \leq n$.

This analysis extends inductively to all higher levels of recursion. Specifically, for deletions in $H^+_i$ at the $i$-th level, type-1 deletions can be handled analogously to the previous levels, while the cost of type-2 deletions is fully covered by the partial credits assigned to the new insertions and re-insertions triggered by deletions in the preceding levels. It is important to note that for any $i$-th recursion level in $\calD^*(H)$, the insertions (both new and re-insertions) used in the above charging scheme occur within the same time interval $\xi$ associated with that level. Their partial credits are exclusively used to pay for type-2 deletions of $H^+_i$ within $\xi$. Therefore, no partial credit is ever charged more than once.

To summarize, each deletion triggers an amortized $O(\log n \log\log n)$ number of re-insertions. Since each insertion costs $O(\log^2 n)$ amortized time, we conclude that the amortized deletion time is $O(\log^3 n \log\log n)$.

The following theorem summarizes our result. 

\begin{theorem}\label{theo:plane}
We can maintain a dynamic set of $n$ planes in $\bbR^3$, with $O(\log^2 n)$ amortized insertion time and $O(\log^3n\log\log n)$ amortized deletion time, so that the lowest plane intersecting a query vertical line can be computed in $O(\log^3 n/\log\log n)$ time.
\end{theorem}

The space of our data structure is $O(n\log n)$ if we store the conflict lists of all shallow cuttings explicitly. Like Chan's original data structure, the space can be improved to $O(n)$ if we store the conflict lists implicitly by using a linear-space static data structure for halfspace range reporting~\cite{ref:AfshaniOp09}; see \cite{ref:ChanA10,ref:KaplanDy20} for the detailed explanations (Chan credited the idea to Afshani).

\paragraph{Non-vertical ray-shooting queries.}
Given a query line $\ell$ that is not necessarily vertical, the goal is to find the bottommost (or topmost) intersection between $\ell$ and $\le(H)$, the lower envelope of $H$. As discussed in~\cite{ref:ChanA10,ref:ChanDy20}, Chan's data structure for vertical ray-shooting queries can be slightly modified to support this type of non-vertical ray-shooting queries. For simplicity, we describe how to adapt the version of Chan’s data structure reviewed in Section~\ref{sec:chanreview}; our data structure can be modified in an analogous manner.

In the preprocessing algorithm {\tt Chan-Preprocess($H$)}, 
for each substructure $D$ of $\calD(H)$, we now explicitly compute the lower envelope of the live planes in $\live(D)$ and build a static data structure to support non-vertical ray-shooting queries on this envelope. Using the algorithm of Dobkin and Kirkpatrick~\cite{ref:DobkinDe90}, such a data structure can be constructed in $O(n\log n)$ time, and each query can be answered in $O(\log n)$ time. Since there are $O(\log n)$ substructures, the total preprocessing time remains $O(n\log n)$, which does not affect the overall time complexity of {\tt Chan-Preprocess($H$)}.

To process a query line $\ell$, we perform a non-vertical ray-shooting query in each substructure to identify the plane on the local envelope that intersects $\ell$. If the corresponding plane has not been deleted, we keep it as a \emph{candidate plane}. This yields $O(\log n)$ candidate planes; we return the lowest one intersecting $\ell$ as the tentative answer. Finally, we verify whether the intersection point actually lies on $\le(H)$, which can be done by using a vertical ray-shooting query. The overall query time remains $O(\log^2 n)$.
As discussed in~\cite{ref:ChanA10}, the original correctness proof for vertical ray-shooting still applies to this setting. 

Our new data structure can be adapted similarly. All time complexity analyses carry over, and the bounds stated in Theorem~\ref{theo:plane} remain valid. For convenience, we refer to this extension as the \emph{non-vertical version of Theorem~\ref{theo:plane}}.

The non-vertical ray-shooting queries are dual to {\em gift wrapping queries} among a dynamic set of 3D points (i.e., find the two tangent planes of the convex hull containing a query line outside the convex hull). 

\subsection{Dynamic lower envelopes for surfaces}
\label{sec:surface}
Our techniques can be extended to the setting of general surfaces in $\mathbb{R}^3$. Let $F$ be a set of $n$ totally defined continuous bivariate functions of constant description complexity in 3D, such that the lower envelope of any subset of functions in $F$ has linear complexity. The goal is to construct a data structure that maintains $F$ under insertions and deletions, and supports efficient vertical ray-shooting queries, i.e., given a vertical query line $\ell$, compute the intersection between $\ell$ and the lower envelope $\le(F)$.

Kaplan, Mulzer, Roditty, Seiferth, and Sharir~\cite{ref:KaplanDy20} employed the same algorithmic scheme as Chan's for the planes, but replaced the vertical shallow cutting algorithm for planes with a new (randomized) vertical shallow cutting algorithm for the surfaces in $F$. Vertical shallow cuttings for $F$ are defined similarly to those for planes; a notable difference is that each cell becomes a (semi-unboudned) ``pseudo-prism''. It is more challenging to compute shallow cuttings for $F$ than for planes. The authors~\cite{ref:KaplanDy20} proposed a randomized algorithm for that. Liu~\cite{ref:LiuNe22} subsequently developed an improved randomized algorithm that can compute a hierarchy of vertical shallow cuttings in expected $O(n\log n)$ time, which resembles the (deterministic) algorithm of Chan and Tsakalidis~\cite{ref:ChanOp16} for planes. Using Liu's shallow cutting algorithm and following Chan's framework, $F$ can be maintained dynamically with $O(\log^2 n)$ amortized expected insertion time and $O(\log^4 n)$ amortized expected deletion time, so that each vertical ray-shooting query can be answered in $O(\log^2 n)$ worst-case deterministic time. By switching to our new algorithmic framework, we can obtain the following result. 

\begin{theorem}\label{theo:surface}
The lower envelope of a set $F$ of $n$ totally defined continuous bivariate functions of constant description complexity in $\bbR^3$, such that the lower envelope of any subset of the functions has linear complexity, can be maintained dynamically with $O(\log^2 n)$ amortized expected insertion time and $O(\log^3 n\log\log n)$ amortized expected deletion time, so that the lowest function intersecting a query vertical line can be found in $O(\log^3 n/\log\log n)$ worst-case deterministic time.
\end{theorem}

\subsection{Applications}
\label{sec:app}

In this section, we discuss several immediate applications of Theorems~\ref{theo:plane} and~\ref{theo:surface}. Most of these are straightforward replacements of earlier results that used the previous data structures~\cite{ref:ChanA10,ref:ChanDy20,ref:KaplanDy20,ref:LiuNe22} with our new results in the two theorems. In general, applying Theorem~\ref{theo:plane} yields deterministic results, while applying Theorem~\ref{theo:surface} yields randomized results.

%Note that this application list is not supposed to be exhaustive. Indeed, we only pick applications where our new results lead to improvements of roughly a logarithmic factor. 

\paragraph{Nearest neighbor queries in a dynamic set of 2D points.}
The problem is to maintain a dynamic set $S$ of $n$ points in $\bbR^2$ under insertions and deletions, while supporting nearest neighbor queries: given a query point $q$, return the point in $S$ closest to $q$. Using a well-known lifting transformation~\cite{ref:deBergCo08}, this problem under the Euclidean distance reduces to the setting of Theorem~\ref{theo:plane}, and thus inherits the same time complexities. For general distance functions, such as the $L_p$ metric for any $p \in [1, \infty]$ or the additively weighted Euclidean distance, the problem reduces, as discussed in~\cite{ref:KaplanDy20}, to the setting of Theorem~\ref{theo:surface}. Consequently, we obtain the same time bounds as in that theorem.

%Note that here and in the following we assume the general distance metric has the property that the functions de

\paragraph{Maintaining dynamic 2D bichromatic closest pair (BCP).}
In this problem, we maintain two dynamic point sets $P$ and $Q$ in $\bbR^2$, each supporting insertions and deletions, with the goal of maintaining the closest pair $(p, q)$ where $p \in P$ and $q \in Q$. Chan~\cite{ref:ChanDy20} proposed a data structure consisting of a heap along with two dynamic lower envelope structures $\calD(P)$ and $\calD(Q)$, achieving $O(\log^2 n)$ amortized insertion time and $O(\log^4 n)$ amortized deletion time, while reporting the closest pair in $O(1)$ time after each update. By replacing $\calD(P)$ and $\calD(Q)$ with our new data structures $\calD^*(P)$ and $\calD^*(Q)$ in a black-box manner, we improve the deletion time to $O(\log^3 n \log \log n)$ while preserving the insertion time and query time. If distances are measured using the Euclidean metric, we obtain a deterministic result via Theorem~\ref{theo:plane}. For general distance metrics, we obtain a randomized result via Theorem~\ref{theo:surface}.

As discussed in \cite{ref:ChanDy20}, the same time complexities can also be achieved for a symmetric problem of maintaining the diameter of a dynamic set of 2D points. 

\paragraph{Computing convex layers of 3D points.}
Given a set $S$ of $n$ points in $\bbR^3$, the problem is to compute all convex layers of $S$. In two dimensions, this problem can be solved in $O(n \log n)$ time~\cite{ref:ChazelleOn85}. For the 3D version, Agarwal and Matou\v{s}ek~\cite{ref:AgarwalDy95} proposed an algorithm whose runtime is dominated by maintaining a dynamic set of planes in $\bbR^3$ under $O(n)$ insertions, deletions, and non-vertical ray-shooting queries. Using Chan’s dynamic data structure~\cite{ref:ChanDy20}, this yields an $O(n \log^4 n)$ time algorithm. By applying our non-vertical variant of Theorem~\ref{theo:plane}, we obtain an improved algorithm with runtime $O(n \log^3 n \log \log n)$.

\paragraph{Maintaining the width of a semi-dynamic 2D point set.}
The problem is to maintain the width of a semi-dynamic (insertion or deletion only) point set in 2D. Eppstein~\cite{ref:EppsteinIn00} reduced each insertion or deletion to an amortized number of $O(\log n)$ insertions, deletions, and vertical ray-shooting queries in a dynamic set of planes in $\bbR^3$. Using Chan’s data structure~\cite{ref:ChanDy20}, each update can be handled in $O(\log^4 n)$ time, which is dominated by the deletion time. By applying our new result in Theorem~\ref{theo:plane}, we can improve the amortized update time to $O(\log^3 n \log\log n)$ for maintaining the width.

\paragraph{Maintaining all-nearest-neighbors graph of a dynamic 2D point set.}
For a 2D point set $S$, its \emph{all-nearest-neighbors graph} has $S$ as its vertex set, and there is a directed edge from $p$ to $q$ if $q$ is the nearest neighbor of $p$ in $S$. The problem is to maintain this graph as $S$ changes due to insertions and deletions. Chan~\cite{ref:ChanSe03} reduces the problem to vertical ray-shooting queries on the lower envelope of a dynamic set of planes in $\bbR^3$, so that each update to the graph can be handled within the update time of the dynamic lower envelope. Thus, using Chan's algorithm~\cite{ref:ChanDy20}, the graph can be maintained in $O(\log^4 n)$ amortized time per update. By applying Theorem~\ref{theo:plane}, we obtain an improved amortized update time of $O(\log^3 n \log\log n)$.

\paragraph{Maintaining geometric MST of a dynamic 2D point set.}
The problem is to maintain a geometric minimum spanning tree (GMST) for a dynamic set $S$ of points in the plane, under insertions and deletions, depending on the chosen distance metric. Following Eppstein~\cite{ref:EppsteinDy95}, the problem reduces to dynamic BCP problem. Specifically, suppose that we have a dynamic 2D BCP data structure with $t_{bcp}$ update time and a dynamic data structure for maintaining MST of general graphs with $t_{mst}$ edge update time (i.e., for insertions or deletions of edges); then they can be combined to yield a dynamic GMST data structure with $t_{gmst}=O((t_{bcp}+t_{mst})\log^2 n)$ amortized update time. 

The current best bound for $t_{mst}$ is $O(\log^4n/\log\log n)$~\cite{ref:HolmFa15}. If the Euclidean distance is used, Chan's data structure~\cite{ref:ChanDy20} gives $t_{bcp}=O(\log^4 n)$, which leads to $t_{gmst}=O(\log^6 n)$. By applying our improved dynamic Euclidean BCP data structure with $t_{bcp}=O(\log^3 n\log\log n)$, we improve this to $t_{gmst}=O(\log^6n/\log\log n)$. 
If a general distance metric is used (e.g., $L_p$ or additively-weighted Euclidean), previous work~\cite{ref:KaplanDy20,ref:LiuNe22} achieves $t_{bcp}=O(\log^4 n)$ expected time, resulting in $t_{gmst}=O(\log^6 n)$ expected time. With our new result, we have $t_{bcp}=O(\log^3 n\log\log n)$ and $t_{gmst}=O(\log^6 n/\log\log n)$ expected time.

If only deletions are allowed, then $t_{mst}$ can be bounded by $O(\log^2n)$~\cite{ref:HolmPo01}. Following the above analysis and using our new result, we have $t_{gmst}=O(\log^5 n\log\log n)$, which improves the previous bound of $O(\log^6n)$ by roughly a logarithmic factor. 

\paragraph{Computing minimum Euclidean bichromatic matching.}
Given two point sets $P$ and $Q$ in $\bbR^2$, each of size $n$, a \emph{minimum Euclidean bichromatic matching} is a set $M$ of $n$ line segments, each connecting a point in $P$ to a point in $Q$, such that the total length of the segments in $M$ is minimized. This problem can be reduced to $O(n^2)$ updates to a dynamic bichromatic closest pair (BCP) data structure under the additively-weighted Euclidean distance~\cite{ref:AgarwalVe99,ref:VaidyaGe89}. Using the previous result with an expected $O(\log^4 n)$ update time for dynamic BCP~\cite{ref:KaplanDy20,ref:LiuNe22}, an algorithm with expected runtime $O(n^2 \log^4 n)$ can be obtained. By applying our new dynamic BCP result with expected $O(\log^3 n \log\log n)$ update time, we obtain an improved algorithm with expected runtime $O(n^2 \log^3 n \log\log n)$.

\paragraph{Dynamic disk graph connectivity.}
Let $S$ be a set of disks in the plane. The \emph{disk graph} $G(S)$ is defined with $S$ as its vertex set, where two disks are connected by an edge if they intersect. A \emph{connectivity query} asks whether two given query disks are connected by a path in $G(S)$. The problem is to maintain $G(S)$ under insertions and deletions of disks in $S$, while supporting efficient connectivity queries.

If only insertions are allowed, Kaplan, Klost, Knorr, Mulzer, and Roditty~\cite{ref:KaplanIn24} obtained the following result. Suppose we have a dynamic nearest neighbor data structure for 2D points under additively-weighted Euclidean metric, with $t_i$ insertion time, $t_d$ deletion time, and $t_q$ query time. Then, each connection query can be answered in $O(\alpha(n))$ time, where $\alpha(n)$ is the inverse Ackermann function, and each insertion can be handled in $T=O(t_d\log^2 n+(t_i+t_q)\log n)$ time. Using the previous result with $t_d=O(\log^4 n)$, $t_i=O(\log^2 n)$, and $t_q=O(\log^2 n)$~\cite{ref:KaplanDy20,ref:LiuNe22}, $T=O(\log^6 n)$ expected insertion time is achieved in \cite{ref:KaplanIn24}. Using our new dynamic nearest neighbor data structure with $t_d=O(\log^3 n\log\log n)$, $t_i=O(\log^2 n)$, and $t_q=O(\log^3 n/\log\log n)$, we can have improved $T=O(\log^5n\log\log n)$ expected insertion time. 

If both insertions and deletions are allowed, Baumann, Kaplan, Klost, Knorr, Mulzer, Roditty, and Seiferth~\cite{ref:BaumannDy24} obtained the following result, assuming that the radius of every disk is in $[1,\Psi]$. Suppose we define $t_i$, $t_d$, and $t_q$ in the same way as above, then 
%have a dynamic nearest neighbor data structure for 2D points under additively-weighted Euclidean metric, with $t_i$ insertion time, $t_d$ deletion time, and $t_q$ query time. Then, 
each connectivity query can be answered in $O(\log n/\log\log n)$ time, and each update can be handled in $T=O(\Psi\cdot (t_i+t_d+t_q))$ time. As above, using our new dynamic nearest neighbor result, we can achieve $T=O(\Psi\log^3n\log\log n)$ expected time, which improves the previous $O(\Psi\log^4 n)$ expected time~\cite{ref:BaumannDy24} by roughly a logarithmic factor.

\paragraph{Spanners for disk graphs.}
Given a set $S$ of $n$ disks in the plane and a parameter $\epsilon > 0$, a $(1+\epsilon)$-spanner is a subgraph $G_{\epsilon}(S)$ of the disk graph $G(S)$ such that for any two disks in $S$, their shortest path distance in $G_{\epsilon}(S)$ is at most $(1+\epsilon)$ times their shortest path distance in $G(S)$. An algorithm for computing such a spanner is given by in \cite{ref:KaplanDy20}, with a runtime of
$O(n \log n + t_i \cdot n \log n + t_d \cdot n / \epsilon^2)$,
where $t_i$ and $t_d$ are the insertion and deletion times, respectively, for a dynamic nearest neighbor data structure under the additively-weighted Euclidean metric. Using the previous data structure~\cite{ref:KaplanDy20,ref:LiuNe22} with $t_i = O(\log^2 n)$ and $t_d = O(\log^4 n)$, the algorithm runs in $O(n / \epsilon^2 \cdot \log^4 n)$ expected time. By replacing the data structure with our new one, which has $t_i = O(\log^2 n)$ and $t_d = O(\log^3 n \log\log n)$, we obtain an improved algorithm with expected runtime $O(n / \epsilon^2 \cdot \log^3 n \log\log n)$.

% \paragraph{Maintaining approximated maximum independent set for dynamic disk graphs.}
% Bhore and Chan~\cite{ref:BhoreDy25} gave an algorithm to maintain an $O(1)$-approximation maximum independent set in the disk graphs of a dynamic set of disks and each update takes $O(\log^6 n)$ amortized (deterministic) time. The algorithm uses Chan's lower envelope data structure 

\section{Shortest paths in weighted disk graphs}
\label{sec:spdisk}

Let $P$ be a set of $n$ points in the plane such that each point $p \in P$ is associated with a weight $r_p$ representing the radius of a disk $D_p$ centered at $p$. Let $S$ be the set of all these disks. 
Given a parameter $r \geq 0$, let $G_r(S)$ denote the \emph{proximity graph} whose vertex set is $P$, and where two vertices $p$ and $p'$ are connected by an edge if $\lVert p - p' \rVert - r_p - r_{p'} \leq r$, where $\lVert p - p' \rVert$ denotes the Euclidean distance between $p$ and $p'$.

Equivalently, if we define the distance between two disks $D_p$ and $D_{p'}$ as $d(D_p, D_{p'}) = \max\{0, \lVert p - p' \rVert - r_p - r_{p'}\}$, then $p$ and $p'$ have an edge in $G_r(S)$ if and only if $d(D_p, D_{p'}) \leq r$. When $r = 0$, $G_r(S)$ becomes the \emph{disk graph} of $P$, i.e., two vertices have an edge if their associated disks intersect. Hence, disk graphs are a special case of proximity graphs.

There are two ways to define edge weights in $G_r(S)$: the \emph{disk-distance case}, where the weight of an edge $(p, p')$ is $d(D_p, D_{p'})$, and the \emph{center-distance case}, where the weight is $\lVert p - p' \rVert$. For the special case $r = 0$, only the center-distance version is meaningful.
For both cases, we consider the \emph{single-source shortest path} (SSSP) problem: given $P$, $r$, and a source point $s \in P$, compute a shortest path tree from $s$ in $G_r(S)$.

In the disk-distance case, Kaplan, Katz, Saban, and Sharir~\cite{ref:KaplanTh23} gave an algorithm with expected runtime $O(n \log^4 n)$, dominated by a dynamic 2D bichromatic closest pair (BCP) data structure with $O(n)$ insertions and deletions. By replacing their BCP data structure with our new one in a black-box manner, we immediately improve the expected runtime to $O(n \log^3 n \log\log n)$.

%A dynamic BCP data structure is also used in \cite{ref:KaplanTh23}. Replacing it with our new BCP data structure gives an algorithm of $O(n\log^5\log\log n)$ time, which is worse than the algorithm of \cite{ref:AnSi25}. 

We now focus on the center-distance case. The authors of \cite{ref:KaplanTh23} also gave an algorithm with expected runtime $O(n \log^6 n)$ for this case. Very recently, An, Oh, and Xue~\cite{ref:AnSi25} in SoCG 2025 developed a new algorithm with runtime $O(n \log^4 n)$ for the special case of disk graphs (i.e., when $r = 0$).
We propose a new algorithm with expected runtime $O(n \log^3 n \log\log n)$, using our dynamic data structure for additively-weighted (Euclidean) nearest neighbor queries. Notably, our algorithm is not derived by merely replacing any component of the algorithm in \cite{ref:AnSi25} in a black-box manner; rather, we introduce an entirely new algorithmic framework. This framework is conceptually simple and may be of independent interest for addressing other variants of the shortest path problem in disk graphs or more broadly in geometric intersection graphs. In addition, our algorithm works for any $r\geq 0$ while the algorithm in \cite{ref:AnSi25} is particularly for the special case $r=0$.  

We describe the algorithm and argue its correctness in Section~\ref{sec:description}, and discuss its implementation and time complexity in Section~\ref{sec:imple}.

\subsection{Algorithm description and correctness}
\label{sec:description}

For each point $p\in P$, let $d_G(p)$ denote the shortest path lengths from $s$ to $p$ in $G_r(S)$. We will only discuss how to compute $d_G(p)$ for all points $p\in P$. The algorithm can be easily modified to also compute the predecessor information of shortest paths. 

For two points $p,q\in P$, we say that they are {\em adjacent} to each other in $G_r(S)$ if they have an edge in the graph.

The algorithm maintains a tentative distance value $dis(p)$ for each point $p$, such that $dis(p) = d_G(p)$ when the algorithm terminates. Initially, we set $dis(s)=0$, and $dis(p)=\infty$ for all other points $p$.
We maintain two subsets $A$ and $B$ of $P$, where $A$ consists of the points whose shortest path distances have been correctly computed (i.e., $dis(p) = d_G(p)$ for all $p \in A$), and $B = P \setminus A$. Furthermore, we partition $A$ into two subsets $A_1$ and $A_2$ such that no point in $A_1$ is adjacent to any point in $B$ in $G_r(S)$.
%implying that $A_1$ is useless in computing shortest paths for points of $B$. 
The algorithm maintains an invariant that is stated in the following observation for reference purpose. 

\begin{observation}\label{obser:invariant}
For every point $p\in A$, it holds that $dis(p)=d_G(p)$. Furthermore, if $p\in A_1$, then for any point $q$ adjacent to $p$ in $G_r(S)$, we must have $q \in A$.
\end{observation}

Initially, we set $A_2=A=\{s\}$, $A_1=\emptyset$, and $B=P\setminus\{s\}$. The algorithm proceeds iteratively (see Algorithm~\ref{algo:spdisk} for the pseudocode). 
In each iteration, we select the point $a\in A_2$ that minimizes $dis(a)+r_a$. Then, we perform the following three {\em main steps}: (1) Find the set $B_{a}$ of points of $B$ adjacent to  $a$ in $G_r(S)$; (2) for each point $b\in B_a$, among all points of $A$ adjacent to $b$ in $G_r(S)$, find the one $q_b$ minimizing $dis(q_b)+\lVert q_b-b\rVert$, and set $dis(b)=dis(q_b)+\lVert q_b-b\rVert$; (3) move all points of $B_a$ from $B$ to $A_2$ and move $a$ from $A_2$ to $A_1$. This completes one iteration of the algorithm. Lemma~\ref{lem:invariant} (see below) proves that for every $b\in B_a$, the computed $dis(b)$ equals $d_G(p)$. This, together with step (3), ensures the invariant in Observation~\ref{obser:invariant} is preserved, thereby establishing the correctness of the algorithm. 

The algorithm terminates once $A_2$ becomes empty, after which if $B\neq \emptyset$, then points of $B$ cannot be reached from $s$ in $G_r(S)$ by Observation~\ref{obser:invariant} and thus $d_G(p)=\infty$ for all $p\in B$.

\begin{algorithm}[h]
	\caption{SSSP algorithm}
	\label{algo:spdisk}
	\SetAlgoNoLine
	\KwIn{$P$, $r$, and a source point $s\in P$}
	\KwOut{$dis(p)$ for all $p\in P$} 
        \BlankLine
	%\tcc{Chan's deletion algorithm but using our new insertions.}
        $A_1\leftarrow \emptyset$, $A_2\leftarrow\{s\}$, $B\leftarrow P\setminus \{s\}$, $dis(s)\leftarrow 0$\;
        \While{$A_2\neq \emptyset$}
        {
          $a\leftarrow \argmin_{p\in A_2}(dis(p)+r_p)$\;
          $B_a\leftarrow\{p\in B: (a,p) \text{ is an edge in } G_r(S)\}$\;
        \For{\em each point $b\in B_a$}
        {
          $q_b\leftarrow\argmin_{p\in A,\ (p,b)\text{ is an edge in } G_r(S)}(dis(p)+\lVert p-b\rVert)$\;
          $dis(b)\leftarrow dis(q_b)+\lVert q_b-b\rVert$\;
        }
        move all points of $B_a$ from $B$ to $A_2$\;
        move $a$ from $A_2$ to $A_1$\;
        }
\end{algorithm}

\begin{lemma}\label{lem:invariant}
$d_G(b)=dis(q_b)+\lVert q_b-b\rVert$ holds for each point $b\in B_a$.
\end{lemma}
\begin{proof}
Consider a point $b\in B_a$. Our goal is to prove $d_G(b)=dis(q_b)+\lVert q_b-b\rVert$. 

Assume to the contrary that $d_G(b)\neq dis(q_b)+\lVert q_b-b\rVert$. By the definition of $q_b$, $q_b$ is adjacent to $b$ in $G_r(S)$. Hence, $d_G(b)\leq dis(q_b)+\lVert q_b-b\rVert$. Since $d_G(b)\neq dis(q_b)+\lVert q_b-b\rVert$, we obtain 
\begin{equation}\label{equ:contradict}
    d_G(b)< dis(q_b)+\lVert q_b-b\rVert.
\end{equation}

Let $\pi(s,b)$ be a shortest path from $s$ to $b$ in $G_r(S)$. Let $p$ be the first point of $\pi(s,b)$ not in $A$. Since $b\not\in A$, such a point $p$ must exist. Let $q$ be the predecessor of $p$ in $\pi(s,b)$; see Figure~\ref{fig:path}. Hence, $q$ has an edge with $p$ in $G_r(S)$ and $q\in A$. 
Depending on whether $p$ is $b$, there are two cases. 

\paragraph{The case $\boldsymbol{p=b}$.}
If $p=b$, then $d_G(b)=d_G(q)+\lVert q-b\rVert$. By the definition of $q_b$, we have $dis(q_b)+\lVert q_b-b\rVert\leq dis(q)+\lVert q-b\rVert$. Since $q\in A$, according to our algorithm invariant in Observation~\ref{obser:invariant}, $dis(q)=d_G(q)$. We can now derive 
\begin{equation*}
\begin{split}
    d_G(b)& =d_G(q)+\lVert q-b\rVert 
           =dis(q)+\lVert q-b\rVert 
           \geq dis(q_b)+\lVert q_b-b\rVert,
\end{split}
\end{equation*}
which contradicts with \eqref{equ:contradict}.

\paragraph{The case $\boldsymbol{p\neq b}$.}
If $p\neq b$, depending on whether $q$ is in $A_1$ or $A_2$ (recall that $q\in A$), there are further two subcases. 
%has been processed, i.e., whether $p'$ has been picked as $a$, the point of $A_2$ with minimum $dis(a)$, there are two cases. 

If $q\in A_1$, then according to our algorithm invariant in Observation~\ref{obser:invariant}, all vertices of $G_r(S)$ adjacent to $q$ must have been moved to $A$. As $q$ is adjacent to $p$ in $G_r(S)$, $p$ must be in $A$, which contradicts with the fact that $p\in B$. Therefore, $q\in A_1$ is not possible.

\begin{figure}[t]
\begin{minipage}[t]{\linewidth}
\begin{center}
\includegraphics[totalheight=0.8in]{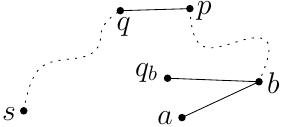}
\caption{The three solid segments represent edges in $G_r(S)$. The two dotted curves represent two subpaths of $\pi(s,b)$.}
\label{fig:path}
\end{center}
\end{minipage}
%\vspace*{-0.15in}
\end{figure}

In the following, we prove the subcase $q\in A_2$. By the definition of $a$, since $q\in A_2$, we have 
\begin{equation}\label{equ:invariant20}
    dis(a)+r_a\leq dis(q)+r_q.
\end{equation}

Recall that $a$ is adjacent to $b$ in $G_r(S)$ and thus $\lVert a-b\rVert-r_a-r_b\leq r$; see Figure~\ref{fig:path}. Since $p\neq b$ and $q$ is the predecessor of $p$ in $\pi(s,b)$, $q$ cannot be adjacent to $b$ in $G_r(S)$ since otherwise $q$ would connect to $b$ directly in $\pi(s,b)$. Hence, $\lVert q-b\rVert-r_q-r_b> r$. We thus obtain $\lVert a-b\rVert-r_a-r_b<\lVert q-b\rVert-r_q-r_b$, which is equivalent to
\begin{equation}\label{equ:invariant30}
     \lVert a-b\rVert-r_a < \lVert q-b\rVert-r_q.
\end{equation}

Combining the above two inequalities \eqref{equ:invariant20} and \eqref{equ:invariant30} leads to 
\begin{equation*}
    dis(a)+\lVert a-b\rVert<dis(q)+\lVert q-b\rVert.
\end{equation*}
By the definition of $q_b$, it holds that $dis(q_b)+\lVert q_b-b\rVert\leq dis(a)+\lVert a-b\rVert$. Hence, $dis(q_b)+\lVert q_b-b\rVert < dis(q)+\lVert q-b\rVert$. Since both $q_b$ and $q$ are in $A$, by Observation~\ref{obser:invariant}, $dis(q_b)=d_G(q_b)$ and $dis(q)=d_G(q)$. Therefore, 
\begin{equation}\label{equ:invariant40}
    d_G(q_b)+\lVert q_b-b\rVert < d_G(q)+\lVert q-b\rVert.
\end{equation}

On the other hand, the length of $\pi(s,b)$, which is $d_G(b)$, is equal to the length of the subpath of $\pi(s,b)$ from $s$ to $q$, which is $d_G(q)$, plus the length of the subpath of $\pi(s,b)$ between $q$ and $b$. By triangle inequality, the length of the latter subpath is larger than or equal to $\lVert q-b\rVert$. Hence,
\begin{equation}\label{equ:invariant50}
    d_G(b)\geq d_G(q)+\lVert q-b\rVert.
\end{equation}

Combining \eqref{equ:invariant40} and \eqref{equ:invariant50} yields $d_G(b)>d_G(q_b)+\lVert q_b-b\rVert$. But this contradicts with \eqref{equ:contradict}. 

The lemma thus follows. 
\end{proof}

\subsection{Algorithm implementation and time analysis}
\label{sec:imple}

We now elaborate on the efficient implementation of the algorithm, particularly the three main steps in each iteration. We use a min-heap to maintain the set $A_2$ with $dis(p)+r_p$ as keys for the points $p\in A_2$, so that the point $a$ can be found by a delete-min operation in $O(\log n)$ time. Each insertion into $A_2$ can be done in $O(\log n)$ time. 

In the following, we focus on the first two main steps. 
We will show that the first main step can be implemented in $O((|B_a|+1)\log^3 n\log\log n)$ expected time using our dynamic additively-weighted (Euclidean) nearest neighbor queries, and the second main step can be implemented in $O(|B_a|\log^3 n)$ time. 
As each point enters $B_a$ at most once throughout the algorithm, the total expected runtime over all iterations is $O(n\log^3 n\log\log n)$.

\subsubsection{The first main step}
The goal of the first main step is to compute $B_a$. By definition, a point $p\in B$ is adjacent to $a$ in $G_r(S)$ if $\lVert p-a\rVert-r_p-r_a\leq r$, which is equivalent to  $\lVert p-a\rVert-r_p\leq r_a+r$. 

We reduce the problem to a sequence of dynamic nearest neighbor queries under additively-weighted Euclidean metric. For each point $p\in B$, we define an additive weight $w(p)=-r_p$, so that the {\em weighted distance} between a point $q\in \bbR^2$ and $p$ is $\lVert p-q\rVert+w(p)$. 

To compute $B_a$, we repeatedly find the additively-weighted nearest neighbor $p\in B$ of $a$, add $p$ to $B_a$, and delete $p$ from $B$, until $\lVert p-a\rVert+w_p> r_a+r$. We maintain the set $B$ using our dynamic additively-weighted nearest neighbor query data structure described in Section~\ref{sec:app}, supporting deletions only throughout the entire algorithm. Each nearest neighbor query takes $O(\log^3 n/\log\log n)$ time and each deletion takes $O(\log^3 n\log\log n)$ amortized expected time. In this way, computing $B_a$ takes $O((|B_a|+1)\log^3n\log\log n)$ amortized expected time.  

\subsubsection{The second main step}
%let $E_b$ denote the set of points of $A$ that have edges with $b$ in $G_r(S)$. By definition, a point $p\in A$ is in $E_b$ if $\lVert pb\lVert-r_p-r_b\leq r$. Our goal is to find the point $q_b\in E_b$ such that $dis(q_b)+\lVert bq_b\rVert$ is minimized.
The goal of the second main step is to compute $q_b$ from $A$ for all points $b\in B_a$.

We maintain a balanced binary search tree $T_A$, the leaves (from left to right) store points $p\in A$ in increasing order of the values $dis(p)+r_p$. For each internal node $v\in T_A$, let $A_v$ denote the subset of points stored at the leaves in the subtree rooted at $v$. We store $Vor_1(v)$, the additively-weighted (Euclidean) Voronoi diagram for the points of $A_v$ with the weight of each point $p\in A$ as $-r_p$. We also store $Vor_2(v)$, the additively-weighted (Euclidean) Voronoi diagram for $A_v$  with the weight of each point $p\in A$ as $dis(p)$.

For each point $b\in B$, to find $q_b$, our algorithm has the following two procedures. 

\begin{enumerate}
\item 
In the first procedure, we find the leftmost leaf $u_b$ of $T_A$ such that $u_b$ is adjacent to $b$ in $G_r(S)$, i.e., $\lVert u_b-b\lVert-r_{u_b}-r_b\leq r$ (for simplicity, we also use $u_b$ to refer to the point stored at $u_b$). Hence, $q_b$ cannot be to the left of $u_b$ and we only need to consider leaves to the right of $u_b$. We can find $u_b$ in $O(\log^2 n)$ time by using the Voronoi diagram $Vor_1(v)$, as follows. 

Starting from the root of $T_A$, for each node $u$, we do the following. Let $v$ be the left child of $u$. Using $Vor_1(v)$, we find the additively-weighted nearest neighbor $q\in A_v$ of $b$ and check whether $\lVert q-b\lVert -r_q\leq r_b+r$. If yes, we search $v$ recursively; otherwise, we search the right subtree of $u$ recursively. The search will eventually reach a leaf and we return it as $u_b$. Since computing a nearest neighbor using Voronoi diagrams takes $O(\log n)$ time, the total time for computing $u_b$ is $O(\log^2 n)$. 

\item 
In the second procedure, using the standard approach, we identify a set $V_b$ of  $O(\log n)$ nodes of $T_A$ whose leaf ranges form a partition of the leaves to the right of (and including) $u_b$. That is, $\bigcup_{v \in V_b} A_v$ is exactly the set of candidate points for $q_b$.

For each node $v\in V_b$, using $Vor_2(v)$, we find the additively-weighted nearest neighbor $p_v\in A_v$ of $b$. Among all such $p_v$, let $q^*$ be one minimizing $dis(p_v)+\lVert p_v-b\lVert$. The next lemma proves that $q^*$ is $q_b$. Therefore, we return $q^*$. 
The runtime of the second procedure is clearly $O(\log^2 n)$. 
\end{enumerate}

\begin{lemma}
The point $q^*$ is $q_b$, that is, $q^*=\argmin_{p\in A,\ (p,b)\text{ is an edge in } G_r(S)}(dis(p)+\lVert p-b\rVert)$. 
\end{lemma}
\begin{proof}
%For each node $v\in V_b$, let $p_v$ be the additively-weighted nearest neighbor of $b$ in $A_v$. Let $q^*$ be the $p_v$ that minimizes $dis(p_v)+\lVert p_v-b\lVert$ among all $v\in V_b$. Our goal is to prove that $q^*$ is $q_b$.
By the definition of the leaf $u_b$, $q_b$ must be in $A_v$ for some $v\in V_b$. By the definition of $q^*$, to prove $q^*$ is $q_b$, it suffices to show that $q^*$ is adjacent to $b$ in $G_r(S)$, i.e., $\lVert q^*-b\lVert-r_{q^*}-r_b\leq r$, or equivalently $\lVert q^*-b\lVert-r_{q^*}\leq r_b+ r$.

First of all, since leaves $p$ of $T_A$ are sorted in increasing order of $dis(p)+r_p$ from left to right, and $q^*$ is not left of $u_b$, we have 
\begin{equation*}
    dis(u_b)+r_{u_b}\leq dis(q^*)+ r_{q^*}. 
\end{equation*}

By the definition of $q^*$, the following holds:
\begin{equation*}
    dis(u_b)+\lVert u_b-b\rVert \geq dis(q^*)+ \lVert q^*-b\rVert. 
\end{equation*}

Subtracting the two inequalities leads to 
\begin{equation}\label{equ:qb}
   \lVert u_b-b\rVert-r_{u_b}\geq \lVert q^*-b\rVert-r_{q^*}.
\end{equation}

Recall that $\lVert u_b-b\lVert-r_{u_b}-r_b\leq r$, or equivalently $\lVert u_b-b\lVert-r_{u_b}\leq r_b+ r$. With the above inequality~\eqref{equ:qb}, we finally obtain $\lVert q^*-b\lVert-r_{q^*}\leq r_b+ r$. The lemma thus follows. 
\end{proof}

The above discussion shows that each $q_b$ can be computed in $O(\log^2 n)$ time using the tree $T_A$. Since we also move each $b\in B_a$ to $A_2$, and thus to $A$, we must insert $b$ into $T_A$. To support such insertions, we replace the static Voronoi diagrams $Vor_1(v)$ and $Vor_2(v)$ at each node $v \in T_A$ with insertion-only Voronoi diagrams. Using the logarithmic method of Bentley and Saxe~\cite{ref:BentleyDe79}, we can implement insertion-only Voronoi diagrams with $O(\log^2 n)$ amortized insertion time and $O(\log^2 n)$ worst-case query time.

With this change, computing each $q_b$ as described earlier takes $O(\log^3 n)$ time. After setting $dis(b) = dis(q_b) + \|q_b - b\|$, we insert $b$ as a new leaf in $T_A$, using $dis(b) + r_b$ as its key. For each ancestor $v$ of the new leaf, we insert $b$ into both $Vor_1(v)$ and $Vor_2(v)$, requiring $O(\log^2 n)$ amortized time. Therefore, the total amortized cost of inserting $b$ into $T_A$ is $O(\log^3 n)$.\footnote{Note that an insertion may trigger a rebalancing operation at some node $v \in T_A$.
Since a node $v$ is rebalanced only after $\Theta(|A_v|)$ new points have been inserted
into $A_v$ since its last rebalancing, we can charge the rebalancing cost to these
new insertions. Therefore, the amortized insertion time is still bounded by $O(\log^3 n)$.}

To summarize, the second main step of the algorithm can be implemented in $O(|B_a| \log^3 n)$ amortized time.

We conclude with the following result:

\begin{theorem}
Given a set $S$ of disks in the plane, a parameter $r \geq 0$, and a source disk, the shortest paths from the source to all other vertices in the proximity graph $G_r(S)$ can be computed in $O(n \log^3 n \log \log n)$ expected time.
\end{theorem}

Setting $r=0$ obtains the following. 
\begin{corollary}
Given in the plane a set $S$ of disks and a source disk, the shortest paths from the source to all other vertices in the disk graph of $S$ can be computed in $O(n\log^3 n\log\log n)$ expected time.     
\end{corollary}

\section{Dynamic halfspace range reporting and $\boldsymbol{k}$-lowest-planes queries}
\label{sec:rangereport}

We consider a more general query in this section. We wish to maintain a dynamic set of points in $\bbR^3$ (under insertions and deletions) to answer the halfspace range reporting queries: Given a query halfspace, report all points inside it. 

In the dual setting, this corresponds to maintaining a dynamic set $H$ of planes in $\bbR^3$ for the {\em below-point queries}: Given a query point, report all planes below it. As shown by Chan~\cite[Corollary 2.5]{ref:ChanRa00}, this problem can be reduced (via a standard ``guessing'' trick on output size) to the following {\em $k$-lowest-planes queries}: Given a vertical line $\ell$ and a number $k$, report the $k$ lowest planes of $H$ at $\ell$. Notably, the $k$-nearest neighbor queries among a dynamic set of 2D points can be reduced to this problem by a standard lifting transformation~\cite{ref:deBergCo08}. 

In what follows, we focus on solving the dynamic $k$-lowest-planes query problem. We describe two solutions that offer different trade-offs between query and update times.

\subsection{The first solution}
\label{sec:sol1}
Using Chan’s data structure~\cite{ref:ChanA10,ref:ChanDy20} reviewed in Section~\ref{sec:chanreview}, it is possible to achieve query time $O(k\log^2 n)$, with the same update times as before, i.e., $O(\log^2 n)$ amortized insertion time and $O(\log^4 n)$ amortized deletion time. Recall that Chan’s data structure $\calD(H)$ consists of $O(\log n)$ substructures, each of which contains $O(\log n)$ vertical shallow cuttings.

The key observation~\cite{ref:ChanA10} is that if a plane $h$ is among the $k$ lowest planes of $H$ and belongs to a substructure $D$ of $\calD(H)$, then $h$ must appear in the conflict list of the cell intersecting $\ell$ in the shallow cutting of $D$ at level $j_k = \log(cn/k)$, for a suitable constant $c$. At this level, the conflict list contains $O(k)$ planes~\cite{ref:ChanA10}.

In light of this, we can answer the query as follows. For each substructure $D$ of $\calD(H)$, we locate the cell intersecting $\ell$ in the shallow cutting of $D$ at level $j_k$ and retrieve all non-deleted planes in its conflict list. Since each conflict list contains $O(k)$ planes, we obtain $O(k \log n)$ candidate planes in total across all substructures. Among them, we return the $k$ lowest planes intersecting $\ell$. The total query time is thus $O((k + \log n) \log n)$.

We can reduce the deletion time by replacing Chan’s data structure with our new one. This yields $O(\log^2 n)$ amortized insertion time and $O(\log^3 n \log\log n)$ amortized deletion time. Since our new data structure contains $O(\log^2 n / \log\log n)$ substructures, the query time now becomes $O((k + \log n) \log^2 n / \log\log n)$. This result is particularly appealing when one aims to minimize the highest logarithmic factors among insertion, deletion, and query times.

\subsection{The second solution}
In certain applications, one may desire a query time with a multiplicative factor of $k$ as small as possible, and ideally a query time of $O(f(n)+k)$ for some logarithmic function $f(n)$. Such a result is possible~\cite{ref:ChanTh12,ref:BergDy23}. In particular, based on Chan's dynamic lower envelope data structure, de Berg and Staals~\cite{ref:BergDy23} developed a data structure of $O(\log^2 n/\log\log n+k)$ query time, $O(\log^{3+\epsilon} n)$ amortized insertion time, and $O(\log^{5+\epsilon} n)$ amortized deletion time, for any arbitrarily small constant $\epsilon>0$. 

We present a new data structure with $O(\log^3 n/\log\log n+k)$ query time, $O(\log^2 n\log\log n)$ amortized insertion time, and $O(\log^4 n\log\log n)$ amortized deletion time. Comparing to the result of de Berg and Staals, ours has smaller highest logarithmic factors among the update and query times, while still achieving constant multiplicative factor in $k$.
We mostly follow the framework of de Berg and Staals, but with two key differences. First, we use our new dynamic lower envelope data structure in place of Chan’s original structure. Second, we make more efficient use of the shallow cutting algorithm by Chan and Tsakalidis~\cite{ref:ChanOp16} in one of the subroutines (i.e., Lemma~\ref{lem:delonly} below). Since our improvement is not just a black-box replacement of Chan’s algorithm, we outline the details of our approach below. 

Recall that our data structure $\calD^*(H)$ consists of $O(\log n/\log \log n)$ recursion levels, where each level has $O(\log n)$ substructures, and 
each substructure has $O(\log n)$ vertical shallow cuttings. For each cell $\triangle$ of the shallow cuttings, we use a deletion-only data structure in Lemma~\ref{lem:delonly} to maintain the live planes in its conflict list and we refer to the data structure as the {\em conflict list data structure}. 
The lemma slightly improves the result of \cite[Lemma~6]{ref:BergDy23}. 
While our algorithm follows the same high-level approach, the improvement stems from a key observation: the shallow cutting algorithm of Chan and Tsakalidis~\cite{ref:ChanOp16} can be applied more efficiently. For completeness, we present the details in the proof.

\begin{lemma}\label{lem:delonly}%{\em\cite[Lemma~6]{ref:BergDy23}}
For any $r\leq n$, there is a data structure %of $O(n\log r)$ size 
to maintain a set $Q$ of $n$ planes to support $O(r\log r)$ amortized time deletions and $O(\log r + n/r  +k)$ time $k$-lowest-planes queries. If a set $Q$ of $n$ planes is given initially, the data structure can be constructed in $O(n\log r)$ time. 
\end{lemma}
\begin{proof}
We build a sequence of $(n/b^i,cn/b^i)$-vertical shallow cuttings $\Xi_i$ for $i=1,2,\ldots, t$, for some constants $b$ and $c$. We choose $t$ such that $b^t\leq r<b^{t+1}$, which implies that $t=O(\log r)$ and $n/b^t\geq n/r$. Using the algorithm of Chan and Tsakalidis~\cite{ref:ChanOp16}, there exist constants $b$ and $c$ such that computing all these cuttings can be done in $O(n\log r)$ time (improving upon the $O(n\log n)$ time used in \cite[Lemma~6]{ref:BergDy23}). 

To delete a plane $h$, we delete $h$ from the conflict lists of the cells of all shallow cuttings that contain $h$. 
When more than half of the planes from a conflict list are deleted, we rebuild the entire structure. Since the conflict list of each cell has at least $n/b^t\geq n/r$ planes, at least $n/2r$ deletions must occur before a rebuild is triggered. As a rebuild takes $O(n\log r)$ time, the amortized deletion time is $O(r\log r)$ (improving upon the $O(r\log n)$ bound in \cite[Lemma~6]{ref:BergDy23}).

Given a query vertical line $\ell$ and a number $k\leq n$, let $i$ be the integer such that $n/b^{i+1}<k\leq n/b^i$. 
Hence, $k=\Theta(n/b^i)$. 

\begin{itemize}
\item 
If $i=0$, then $k=\Theta(n)$. In this case, we can simply find the $k$ lowest planes at $\ell$ in $O(k)$ time using a linear-time selection algorithm. 

\item 
If $1\leq i\leq t-1$, then we compute the cell $\triangle$ of $\Xi_i$ that intersects $\ell$, which takes $O(\log r)$ time as the number of cells of $\Xi_i$ is $O(b^i)$, which is $O(r)$. By definition, the lowest $n/b^i$ planes at $\ell$ must be in the conflict list of $\triangle$ (because $\Xi_i$ covers the $n/b^i$-level of $Q$). Since $k\leq n/b^i$, the $k$ lowest planes at $\ell$ in the conflict list must be the $k$ lowest planes of $Q$. Therefore, we find the $k$ lowest planes at $\ell$ in the conflict list of $\triangle$ and return them as the answer to the query. Since the conflict list of $\triangle$ has $O(n/b^i)$ planes, the above takes $O(n/b^i)$ time, which is $O(k)$ as $k=\Theta(n/b^i)$.

\item 
If $i\geq t$, then $n/b^t\geq k$. We compute the cell $\triangle$ of the shallow cutting $\Xi_t$ that intersects $\ell$, which takes $O(\log r)$ time. As above, the $k$ lowest planes at $\ell$ in the conflict list of $\triangle$ must be the $k$ lowest planes of $Q$. Therefore, 
we find the $k$ lowest planes at $\ell$ in the conflict list of $\triangle$ and return them as our answer. Since the conflict list size is $O(n/b^t)$, which is $O(n/r)$,
%since $n/b^t\geq n/r$, 
the time for searching the conflict list is $O(n/r)$.
\end{itemize}

Combining the above three cases, the total query time is $O(\log r+n/r+k)$.
\end{proof}

\paragraph{Queries.}
Given a query line $\ell$ and a number $k\leq n$, the query algorithm proceeds as follows. For each substructure $D$ of $\calD^*(H)$, we locate the cell $\triangle$ intersecting $\ell$ in the $j_k$-th level of the shallow cutting of $D$, which takes $O(\log n)$ time. 
Since each conflict list associated with such a cell contains $O(k)$ planes and there are $O(\log^2 n/\log\log n)$ substructures, we gather $O(\log^2 n/\log\log n)$ conflict lists of size $O(k)$ each, in a total of $O(\log^3n/\log\log n)$ time (without explicitly examining the contents of these lists yet). 

Next, we find the $k$ lowest planes among all planes across these conflict lists. For this, we resort to a technique of querying multiple data structures simultaneously in \cite[Theorem~1]{ref:BergDy23}, which is based on an adaption of the heap selection algorithm of Frederickson~\cite{ref:FredericksonAn93}. Suppose we have a data structure for each conflict list that can answer a $k$-lowest-planes query in $O(\tau(m)+k)$ time, where $m$ is the size of the conflict list. Then, the technique can find the $k$ lowest planes among all conflict lists in $O(k+\tau(m)\cdot t)$ time, where $t$ is the number of conflict lists. 

In our setting, $t=O(\log^2 n/\log\log n)$ and $\tau(m)=\log r+m/r$ by Lemma~\ref{lem:delonly}. Furthermore, we set $r=\log^2 n/\log\log n$ when applying Lemma~\ref{lem:delonly}. Since each conflict list size is $m=O(k)$, we get $\tau(m)=O(\log\log n+k\log\log n/\log^2 n)$. 
Plugging these into the above $O(k+\tau(m)\cdot t)$ time complexity, we obtain that the time is $O(k+(\log\log n+k\log\log n/\log^2 n)\cdot \log^2 n/\log\log n)$, which is $O(k+\log^2 n)$. 

In summary, the overall query time is bounded by $O(k+\log^3n/\log\log n)$.

\paragraph{Insertions.}
To insert a plane $h$, we still periodically rebuild the data structure $\calD^*(H)$ as before. The difference now is that, during the rebuild, we also construct the conflict list data structures described in Lemma~\ref{lem:delonly}, which introduces an additional $\log r = O(\log\log n)$ factor in the preprocessing time. Specifically, for a set $H$ of $n$ planes, building $\calD^*(H)$ along with all conflict list data structures requires $O(n \log n \log \log n)$ time.

As a result, the insertion time follows the recurrence $T(n) = T(3n/4) + P(n)/n$, where $P(n) = O(n \log n \log \log n)$ is the preprocessing time. Solving this recurrence yields an amortized insertion time of $O(\log^2 n \log \log n)$.

\paragraph{Deletions.}
To delete a plane $h$, we follow the same algorithm as before. In addition, we also perform deletions from the conflict list data structures. Specifically, recall that each plane is live in exactly one substructure and it appears in at most $O(\log n)$ conflict lists of the substructure (i.e., in the conflict lists of $O(\log n)$ cells of the shallow cuttings of the substructure). Therefore, we need to delete $h$ from $O(\log n)$ conflict list data structures. 
Furthermore, deleting $h$ triggers $O(\log n\log\log n)$ re-insertions. For each of these re-insertions $h'$, we also need to delete $h'$ from $O(\log n)$ conflict list data structures that contain $h'$.  

As a result, the deletion time for $h$ is bounded by $O(t\cdot (I(n)+D(n)\cdot \log n))$, where $t$ is the number of re-insertions, $I(n)$ is the insertion time, and $D(n)$ is the deletion time of the conflict list data structure. We already know that $t=O(\log n\log\log n)$ and $I(n)=O(\log^2 n\log\log n)$. By Lemma~\ref{lem:delonly}, $D(n)=O(r\log r)$, which is $O(\log^2 n)$ since $r=\log^2n/\log\log n$. Hence, the amortized deletion time is bounded by $O(\log^4 n\log\log n)$.

\subsection{Summary, applications, and extensions}

The following theorem summarizes our results in this section. 

\begin{theorem}\label{theo:rangereport}
We can maintain a dynamic set of $n$ planes in $\bbR^3$, with $O(\log^2 n)$ amortized insertion time and $O(\log^3n\log\log n)$ amortized deletion time, so that each $k$-lowest-planes query can be answered in $O((\log n+k)\log^2 n/\log\log n)$ worst-case time; alternatively, we can achieve $O(\log^2 n\log\log n)$ amortized insertion time, $O(\log^4n\log\log n)$ amortized deletion time, and $O(k+\log^3 n/\log\log n)$ worst-case query time.
\end{theorem}

\paragraph{Applications.} Theorem~\ref{theo:rangereport} has many immediate applications.  
As mentioned above, the below-point queries can be reduced to the $k$-lowest-planes queries~\cite[Corollary 2.5]{ref:ChanRa00}. Therefore, we can obtain the same result for the below-point queries as Theorem~\ref{theo:rangereport}, with $k$ in the query time representing the output size. In the dual setting, this is equivalent to maintaining a dynamic set of points in $\bbR^3$ to support halfspace range reporting queries (the title of this section).

Furthermore, we can also solve the $k$-nearest neighbor queries for a dynamic set of 2D points. Indeed, by a standard lifting transformation, the problem reduces to the $k$-lowest-planes queries for a dynamic set of planes in $\bbR^3$, and therefore Theorem~\ref{theo:rangereport} can be applied directly.

In addition, we can also solve the following disk range queries for a dynamic set of 2D points: given a query disk, report all points inside the disk. By the lifting transformation, the problem reduces to the below-point queries for a dynamic set of 3D planes. We note that for the special case where all query disks are unit disks (i.e., have the same known radius), it is possible to achieve an improved query time of $O(k+\log n)$~\cite{ref:WangDy25}.

\paragraph{Extensions to the $\boldsymbol{k}$-lowest-surfaces queries.}
We can also extend our result to support $k$-lowest-surfaces queries for the general problem setting described in Section~\ref{sec:surface}. Building on Chan’s dynamic lower envelope framework and the observation discussed in Section~\ref{sec:sol1}, previous work~\cite{ref:KaplanDy20,ref:LiuNe22} achieved amortized expected insertion and deletion times of $O(\log^2 n)$ and $O(\log^4 n)$, respectively, along with a query time of $O(k \log^2 n)$.

Using our new data structure for Theorem~\ref{theo:surface}, we can obtain the same query and update time bounds as in the first solution of Section~\ref{sec:sol1}, except that both the insertion and deletion times are expected.

Extending the second solution to the surface case presents a technical challenge: it is unclear whether Liu’s algorithm~\cite{ref:LiuNe22} can compute all the required shallow cuttings in the proof of Lemma~\ref{lem:delonly} within $O(n \log r)$ expected time. If it can, then the second solution from Theorem~\ref{theo:rangereport} directly applies to the surface case as well, again with expected update times. Otherwise, Liu’s algorithm can still compute these cuttings in $O(n \log n)$ expected time, which increases the preprocessing time of Lemma~\ref{lem:delonly} to $O(n \log n)$ and the amortized deletion time to $O(r \log n)$. Consequently, the amortized insertion and deletion times increase to $O(\log^3 n)$ and $O(\log^5 n)$, respectively.

As noted in previous work~\cite{ref:KaplanDy20,ref:LiuNe22}, an immediate application of these results is to answer $k$-nearest neighbor queries for a dynamic set of 2D points under general distance functions.

%\footnotesize
\bibliographystyle{plainurl}
\bibliography{reference}

\section*{Appendix: A note on the incorrectness of the main result}
\label{sec:app}

This note concerns an error in this paper recently published in~\cite{ref:WangDy26}. After publication, a flaw was identified in the main algorithm. Attempts to correct the issue have so far been unsuccessful, and as a result, the main result of the paper (Theorems \ref{theo:plane} and \ref{theo:surface}) is not currently supported. The purpose of this note is to document the issue so that readers are aware of it.

The issue arises in the deletion algorithm {\tt Delete($H,h$)} described in Section~\ref{sec:deletion}. After a plane $h$ is deleted using the procedure {\tt Chan-Delete($H^+,h$)}, some planes in $H^+$ may be reinserted into $H^-$. However, these planes may still remain in $H^+$, that is, they may still appear in the conflict lists of certain cutting cells of $\calD(H^+)$. Consequently, it is possible for a plane to belong to both $H^+$ and $H^-$. The deletion algorithm, however, is designed under the assumption that $H^+ \cap H^- = \emptyset$. This violation of the assumption causes the query algorithm to be incorrect.

One possible approach to addressing this issue is to modify Algorithm~\ref{algo:delete} as follows. If a plane $h$ belongs to both $H^-$ and $H^+$, one could first call {\tt Chan-Delete($H^+,h$)} and then  call {\tt Delete($H^-,h$)} recursively. With this modification, the query algorithm would behave correctly. However, the deletion time would increase to $O(\log^4 n)$ (or, at the very least, it is unclear whether a more refined analysis would improve the time bound), matching the deletion time of previous work~\cite{ref:ChanA10,ref:ChanDy20,ref:KaplanDy20,ref:LiuNe22}. As a result, this approach does not preserve the main performance improvement claimed in the paper (Theorems \ref{theo:plane} and \ref{theo:surface}).

\paragraph{Acknowledgment.}
I would like to thank Hanwen Zhang (University of Copenhagen) for pointing out the issue and for helpful discussions regarding possible fixes.

\end{document}